\documentclass[preprint]{elsarticle}
\usepackage{graphicx}
\usepackage{amsmath,amssymb}
\usepackage{extpfeil}
\usepackage{bm}
\usepackage{ifthen}
\usepackage[T1]{fontenc}
\usepackage[utf8]{inputenc}
\usepackage{microtype}
\usepackage{subfig}
\usepackage{graphicx}
\usepackage{tikz}
\usepackage{adjustbox}
\usetikzlibrary{trees}
\usepackage{rotating}
\usepackage{pgfplots}
\usetikzlibrary{automata,positioning}
\usetikzlibrary{shapes.geometric}
\renewcommand{\labelitemi}{$\bullet$}

\tikzstyle{level 1}=[level distance=2.5cm, sibling distance=4.5cm]
\tikzstyle{level 2}=[level distance=2.3cm, sibling distance=4.5cm]
\tikzstyle{term} = [text centered]

\newboolean{lv}
\setboolean{lv}{false}
\newboolean{verbose}
\setboolean{verbose}{false}
\newboolean{monotonicity}
\setboolean{monotonicity}{false}
\newcommand{\bnf}{::=}
\newcommand{\midd}{\; \; \mbox{\Large{$\mid$}}\;\;}

\newcommand{\termone}{M}
\newcommand{\termtwo}{N}
\newcommand{\termthree}{L}
\newcommand{\termfour}{P}
\newcommand{\termfive}{S}
\newcommand{\termsix}{V}
\newcommand{\termseven}{A}

\newcommand{\varone}{x}
\newcommand{\abstr}[2]{\lambda #1.#2}
\newcommand{\subst}[3]{#1\{#2/#3\}}
\newcommand{\rdxone}{R}
\newcommand{\rdxtwo}{Q}
\newcommand{\rdxs}[1]{\mathcal{R}_{#1}}

\newcommand{\contone}{C}
\newcommand{\conttwo}{D}
\newcommand{\contthree}{E}

\newcommand{\redbeta}{\longrightarrow_\beta}
\newcommand{\redbetaclo}{\xtwoheadrightarrow{}_\beta}
\newcommand{\red}{\rightarrow}
\newcommand{\redlo}{\longrightarrow_\pslo}
\newcommand{\coredlo}{{}_\pslo\!\longleftarrow}
\newcommand{\redri}{\longrightarrow_\psri}

\newcommand{\redbetared}[1]{\overset{#1}{\xtwoheadrightarrow{}_\beta}}
\newcommand{\redlosteps}[1]{\longrightarrow_\pslo^#1}

\newcommand{\redbetasteps}[1]{\longrightarrow^{#1}_\beta}

\newcommand{\psuni}{\mathsf{U}}
\newcommand{\pslo}{\mathsf{LO}}
\newcommand{\psri}{\mathsf{RI}}

\newcommand{\dist}[1]{\mathsf{Dist}\left(#1\right)}
\newcommand{\supp}[1]{\mathbf{Supp}\left(#1\right)}
\newcommand{\pdist}[1]{\mathsf{PDist}\left(#1\right)}

\newcommand{\explen}[1]{\mathsf{ExpLen}_{#1}}

\newcommand{\nsteps}[1]{\mathsf{Steps}_{#1}}

\newcommand{\poly}[1]{\mathsf{Poly}(#1)}

\newenvironment{varitemize}
{
\begin{list}{\labelitemi}
{\setlength{\itemsep}{0pt}
 \setlength{\topsep}{0pt}
 \setlength{\parsep}{0pt}
 \setlength{\partopsep}{0pt}
 \setlength{\leftmargin}{15pt}
 \setlength{\rightmargin}{0pt}
 \setlength{\itemindent}{0pt}
 \setlength{\labelsep}{5pt}
 \setlength{\labelwidth}{10pt}
}}
{
 \end{list} 
}

\newcounter{numberone}
\newenvironment{varenumerate}
{
	\begin{list}{\arabic{numberone}.}
		{
			\usecounter{numberone}
			\setlength{\itemsep}{0pt}
			\setlength{\topsep}{0pt}
			\setlength{\parsep}{0pt}
			\setlength{\partopsep}{0pt}
			\setlength{\leftmargin}{15pt}
			\setlength{\rightmargin}{0pt}
			\setlength{\itemindent}{0pt}
			\setlength{\labelsep}{5pt}
			\setlength{\labelwidth}{15pt}
		}}
		{
		\end{list} 
	}

\newdefinition{remark}{Remark}
\newdefinition{example}{Example}
\newtheorem{theorem}{Theorem}
\newtheorem{lemma}{Lemma}
\newtheorem{proposition}{Proposition}
\newtheorem{corollary}{Corollary}
\newdefinition{definition}{Definition}
\newproof{proof}{Proof}
\newdefinition{notation}{Notation}

\begin{document}
\title{On Randomised Strategies in the $\lambda$-Calculus\tnoteref{t1}\tnoteref{t2}}
\tnotetext[t1]{This work was partially supported by ANR grant \textsc{Elica} ANR-14-CE25-0005 and Inria/JSPS EA \textsc{Crecogi}.}
\tnotetext[t2]{A preliminary version of this paper appeared in the Proceedings of the 19th Italian Conference on Theoretical Computer Science, ICTCS 2018 \cite{Conf2018}.}
\author{Ugo Dal Lago}
\ead{ugo.dallago@unibo.it}

\author{Gabriele Vanoni}
\ead{gabriele.vanoni2@unibo.it}

\address{Università di Bologna \& INRIA Sophia Antipolis}

\begin{abstract}
  In this work we study randomised reduction strategies---a notion
  already known in the context of abstract reduction systems---for the
  $\lambda$-calculus. We develop a simple framework that allows us to
  prove a randomised strategy to be positive almost-surely
  normalising. Then we propose a simple example of randomised strategy
  for the $\lambda$-calculus that has such a property and we show why
  it is non-trivial with respect to classical deterministic strategies
  such as leftmost-outermost or rightmost-innermost. We conclude
  studying this strategy for two sub-$\lambda$-calculi, namely those
  where duplication and erasure are syntactically forbidden, showing
  some non-trivial properties.
\end{abstract}


\maketitle 

\section{Introduction}
There are different possible \emph{strategies} you can follow to
evaluate expressions. Some are better than others, and bring you to
the result in a lower number of steps. Since programs in pure
functional languages are essentially expressions, the problem of
defining good strategies is particularly interesting. Finding
\emph{minimal} strategies, i.e. strategies that minimise the number of
steps to normal form, seems even more interesting. However, the problem
of picking the redex leading to the reduction sequence of minimal length
has been proven undecidable for the $\lambda$-calculus
\cite[Section~13.5]{barendregt_lambda_1984}, the core of
pure functional programming languages. In the last decades several reduction
strategies have been developed. Their importance is crucial in the
study of evaluation orders in functional programming languages, and defines an important part of
their semantics. The reader can think about the differences between
\texttt{Haskell} (\emph{call-by-need}) and \texttt{Caml}
(\emph{call-by-value}). The $\lambda$-calculus is a good abstraction
to study reduction mechanisms because of its very simple
structure. In fact, although \emph{Turing-complete}, it can be seen as
a rewriting system \cite{terese_term_2003}, where terms can be formed
in only two ways, by abstraction and application, and only one
rewriting rule, the $\beta$-rule, is present. Reduction strategies for
the $\lambda$-calculus are typically defined according to the position
of the contracted \emph{redex} e.g. \emph{leftmost-outermost},
\emph{leftmost-innermost}, \emph{rightmost-innermost}. As the following
trivial examples show, the adopted strategy can indeed have a strong
impact on evaluation performances, and possibly also on termination
behaviour.
\begin{example}\label{example:canc}
  Let $\bm{I}=\lambda x.x$, $\bm{\omega}=\lambda x.xx$ and
  $\bm{\Omega}=\bm{\omega\omega}$. We now consider the reduction of
  the term $(\lambda x.\bm{I})\bm{\Omega}$ according to two different
  reduction strategies, namely leftmost-outermost ($\pslo$) and
  rightmost-innermost ($\psri$).
  \begin{align*}
  (\lambda x.\bm{I})\bm{\Omega}&\redlo \bm{I}\\
  (\lambda x.\bm{I})\bm{\Omega}&\redri(\lambda x.\bm{I})\bm{\Omega}\redri(\lambda x.\bm{I})\bm{\Omega}\redri\cdots
  \end{align*}
  The term $\bm{\Omega}$ is a looping combinator i.e. it reduces to
  itself. However, in $(\lambda x.\bm{I})\bm{\Omega}$ the argument
  $\bm{\Omega}$ is discarded since the function returns the identity combinator.
  Thus, leftmost-outermost (akin to call-by-name in functional
  programming languages) yields a normal form in one step. Conversely,
  rightmost-innermost (akin to call-by-value) continues to evaluate
  the argument $(\lambda x.\bm{I})\bm{\Omega}$, though it is useless, and
  rewrites always the same term, yielding to a non-terminating
  process.
\end{example}
\begin{example}\label{example:copy}
  We now consider the reduction of the term
  $(\lambda x.xx)(\bm{II})$, according to $\pslo$ and $\psri$
  strategies, as above.
  \begin{align*}
    (\lambda x.xx)(\bm{II})&\redlo (\bm{II})(\bm{II}) \redlo
    \bm{I}(\bm{II}) \redlo \bm{II} \redlo \bm{I}\\
    (\lambda x.xx)(\bm{II})&\redri (\lambda x.xx)\bm{I} \redri \bm{II} \redri \bm{I}
  \end{align*}
  Here the argument $\bm{II}$ is duplicated and thus it is much more convenient
  to reduce it before it is copied, as in
  rightmost-innermost. Leftmost-outermost does, indeed, some useless
  work.
\end{example}
In general, innermost strategies are considered more efficient, because
programs often need to copy their arguments (as in
Example \ref{example:copy}). However, as seen in Example
\ref{example:canc}, rightmost-innermost is not normalising: there
exist terms which have a normal form which, however, can be missed by
innermost strategies. Instead, a classical result by Curry and Feys
\cite{curry_combinatory_1958} states that the leftmost-outermost strategy is
\emph{normalising}, i.e. it always rewrites terms to their normal
norm, if it exists. Thus, leftmost-outermost is slower, but
safer. Could we get, in a sense, the best of both worlds?  All
reduction strategies for the $\lambda$-calculus in the literature up
to now are \emph{deterministic}, i.e. they are (partial) functions on
(possibly shared representations of) terms. There is however some work on probabilistic term
rewriting systems
\cite{bournez_proving_2005,ferrer_fioriti_probabilistic_2015,avanzini_probabilistic_2018},
in particular regarding termination, and about randomised strategies
in the abstract~\cite{bournez_probabilistic_2002}. What would
happen if the redexes to reduce were picked according to some
probability distribution? How many steps would a term need to reach a
normal form \emph{on the average}?

In this work we consider a simple \emph{randomised}
reduction strategy $\mathsf{P}_\varepsilon$, where the $\pslo$-redex
is reduced with probability $\varepsilon$ and the $\psri$-redex is
reduced with probability $1-\varepsilon$. This is not necessarily
the most interesting example, but certainly a good starting point
in our investigation. The \emph{uniform} randomised strategy,
which picks one between \emph{all} the redexes in the term uniformly
at random looks more natural, although much more difficult to analyse:
there is no fixed lower bound on the probability of picking \emph{the standard}
redex, i.e. the leftmost-outermost one. The following are our main results.
\begin{varitemize}
\item
  The uniform strategy is \emph{not} positive almost surely normalising. 
\item
  For every, $0<\varepsilon\leq 1$, the strategy
  $\mathsf{P}_\varepsilon$ is positive almost-surely normalising on
  weakly normalising terms. That means that if a term $\termone$ is
  weakly normalising, then the expected number of reduction steps from
  $\termone$ to its normal form with strategy $\mathsf{P}_\varepsilon$
  is finite. This is in contrast to the rightmost-innermost strategy, as can be
  seen from Example \ref{example:canc}. Rightmost-innermost, in other
  words, is the only non-normalising strategy in the family
  $\{\mathsf{P}_\varepsilon\}_{0\leq\varepsilon\leq 1}$, namely $\mathsf{P}_0$.
\item The function $\explen{\termone}(\varepsilon)$ representing the 
  average number of steps a term $\termone$ needs to reach normal form under strategy
  $\mathsf{P}_\varepsilon$ is a power series.
\item
  The family of strategies
  $\{\mathsf{P}_\varepsilon\}_{0<\varepsilon<1}$ is shown to be
  non-trivial. In other words, there exists a class of terms and
  a real number $0<\mu<1$ for which $\mathsf{P}_\mu$ outperforms, on average, both
  $\pslo$ and $\psri$. This shows that randomisation can
  indeed be useful in this context. This is not surprising:
  in computer science there are many situations in which adding a random factor
  improves performances, e.g. in randomised 
  algorithms~\cite{motwani_randomized_1995}, which are
  often faster (on average) than any (known) deterministic algorithm. 
  Furthermore we show that also the converse can hold, i.e.
  there exists a term for which both $\pslo$ and $\psri$ outperform $\mathsf{P}_\varepsilon$
  for every $0<\varepsilon<1$.
\item
  The expected number of reduction steps to
  normal form with strategy $\mathsf{P}_\varepsilon$, seen as a function
  on $\varepsilon$, has minimum at $1$ for terms
  in the affine $\lambda$-calculus $\lambda A$  and maximum at $1$ for $\lambda 
  I$-terms. However, this function is neither monotonic nor convex nor concave 
  in either $\lambda A$ or $\lambda I$: already in simple calculi it can be 
  very chaotic.
%  Moreover in $\lambda I$
%              [$\lambda A]$, $\mathsf{P}_0$ (i.e. rightmost-innermost)
%              [$\mathsf{P}_1$ (i.e. leftmost-outermost)] minimizes the
%              number of steps to normal form among \emph{all} the
%              possible strategies.
\end{varitemize}
The rest of this paper is structured as follows. In Section
\ref{section:basic}, basic definitions and results for the untyped
$\lambda$-calculus are given. In Section \ref{section:prob} we present
our model of fully probabilistic abstract reduction systems, and we
give a sufficient condition for positive almost-sure termination. In Section
\ref{section:prostra} we apply this model to the $\lambda$-calculus,
defining a randomised reduction strategy and collecting some
results. Section \ref{section:conclusions} concludes the paper
with some ideas for further investigations on the subject.
%%%%%%%%%%%%%%%%%%%%%%%%%%%%%%
\subsection*{Acknowledgements}
%%%%%%%%%%%%%%%%%%%%%%%%%%%%%%
Our interest in randomised strategies comes from a discussion the first author 
had with Prakash Panangaden.
We thank also the anonymous reviewers for their many and insightful comments.
%%%%%%%%%%%%%%%%%%%%%%%%%%%%%%%%%%%%%
\section{Basic Notions and Notations}\label{section:basic}
%%%%%%%%%%%%%%%%%%%%%%%%%%%%%%%%%%%%%%
The following definitions are standard and are adapted from \cite{terese_term_2003}.
\begin{definition}
    Assume a countable infinite set $\mathcal{V}$ of variables. The
    $\lambda$-\emph{calculus} is the language of terms defined by
    the following grammar:
    $$
    \termone,\termtwo\bnf\varone\in\mathcal{V}\midd\termone\termtwo\midd\abstr{\varone}{\termone}
    $$
    We denote by $\Lambda$ the set of all $\lambda$-terms. As usual,
    $\lambda$-terms are taken modulo $\alpha$-equivalence, which allows
    to appropriately define the capture-avoiding substitution of all
    the free occurrences of $\varone$ for $\termtwo$ in $\termone$,
    denoted by $\subst{\termone}{\termtwo}{\varone}$.
\end{definition}
\begin{lemma}[Substitution Lemma]
	For any $\termone,\termtwo,\termthree\in\Lambda$, if $x\neq y$ and $x$ is not free in  $\termthree$, then
	$$
	\termone\{\termtwo/x\}\{\termthree/y\}=\termone\{\termthree/y\}\{\termtwo\{\termthree/y\}/x\}.
	$$
\end{lemma}
Reduction will be defined based on the notion of a context, which
needs to be given a formal status.
\begin{definition}
  We define (one-hole) \emph{contexts} by the following grammar:
  $$
  \contone,\conttwo\bnf\Box\midd\contone\termone\midd\termone\contone\midd\abstr{\varone}{\contone}
  $$
  We denote with $\Lambda_\Box$ the set of all contexts.
\end{definition}
Intuitively, contexts are $\lambda$-terms with a hole that can be
filled with another $\lambda$-term. We indicate with
$\contone[\termone]$ the term obtained by replacing $\Box$ with
$\termone$ in $\contone$, contexts can in fact bind variables. Those $\lambda$-terms in the form
$\rdxone=(\lambda x.\termone)\termtwo$ are called \emph{$\beta$-reducible expressions} or
$\beta$-\emph{redexes} and $\termone\{\termtwo/x\}$ is said to be
the \emph{contractum} of $\rdxone$. This is justified by the following definition.
\begin{definition}
	The relation of $\beta$-\emph{reduction}, $\redbeta\subseteq\Lambda\times\Lambda$, is defined as
	$$
	\redbeta = \{(\contone[(\lambda x.\termone)\termtwo],\, \contone[\termone\{\termtwo/x\}])\, |\, \termone, \termtwo\in\Lambda, \contone\in\Lambda_\Box\}.
	$$
	We denote by $\redbetaclo$ the reflexive and transitive closure of $\redbeta$.
\end{definition}

We introduce some concepts of rewriting theory that will be useful in the following sections. We borrow terminology from \cite{vanoostrom_2007}, except for the fact that we do not label reduction steps.

\begin{definition}[ARS]
	An \emph{abstract reduction system} (ARS) is a pair $(A,\rightarrow)$ where $A$ is a set of \emph{objects} with cardinality at most countable and $\rightarrow\,\subseteq A\times A$ is the \emph{reduction relation}.
\end{definition}

\begin{definition}[Sub-ARS]
	Let $\mathcal{A}=(A,\red_\alpha)$ and $\mathcal{G}=(G,\red_\gamma)$ be two 
	ARSs. $\mathcal{A}$ is a \emph{sub-ARS} of $\mathcal{G}$ if $A\subseteq G$ 
	and $\red_\alpha\subseteq\red_\gamma$.
\end{definition}

If $(a,b)\in\,\rightarrow$ for $a,b\in A$ in an ARS $(A,\red)$ , then we write 
$a\rightarrow b$ and if $a\red a_1\red\cdots\red a_{n-1}\red b$, we write 
$a\rightarrow^n b$. $\twoheadrightarrow$ is the reflexive and transitive 
closure of $\rightarrow$. A \emph{reduction sequence} is a finite or infinite 
sequence $\sigma:a_0\rightarrow a_1\rightarrow\cdots$. If $\sigma$ is finite, 
then $|\sigma|$ is the \emph{length} of $\sigma$. We say that $a\in A$ is in 
\emph{normal form} if there exist no $b\in A$ such that $a\rightarrow b$. We 
call $\mathbf{NF}(A)$ the subset of $A$ whose elements are the normal forms. An 
object $a$ is \emph{(weakly) normalising} if there exists a reduction sequence 
such that $a\twoheadrightarrow b$ and $b$ is in normal form. An object $a$ is 
\emph{strongly normalising} if every reduction sequence from $a$ is finite.

We can see the $\lambda$-calculus defined above as an ARS $(\Lambda,\redbeta)$.
We denote by $\Lambda_\mathsf{WN}$ the set of weakly normalising terms of
$\Lambda$ and by $\Lambda_\mathsf{SN}$ the set of strongly normalising terms of
$\Lambda$.
%%%%%%%%%%%%%%%%%%%%%%%%%%%%%%%%%%%%%%%%
\subsection{Two Subcalculi of $\Lambda$}
%%%%%%%%%%%%%%%%%%%%%%%%%%%%%%%%%%%%%%%%
Full $\lambda$-calculus is very powerful and flexible, but it has a complicated dynamics. Sometimes it is useful to restrict ourselves to a subset of terms. In particular we focus our attention on two subsystems where terms satisfy a predicate on the number of occurrences of free variables. These systems are meaningful because they are \emph{stable} w.r.t. $\beta$-reduction i.e. if
$\termone\in S$ and $\termone\redbeta\termtwo$ then $\termtwo\in
S$. The interested reader can find a comprehensive treatment in 
\cite{sinot_sub-$lambda$-calculi_2008}.
%%%%%%%%%%%%%%%%%%%%%%%%%%%%%%%%%%%%%
\subsubsection{The $\lambda I$-calculus.}
%%%%%%%%%%%%%%%%%%%%%%%%%%%%%%%%%%%%%
The $\lambda I$-calculus was the original calculus studied by Alonzo
Church in 1930s \cite{church_unsolvable_1936}, and
\cite{barendregt_lambda_1984} contains a whole section dedicated to
it. In $\lambda I$-calculus there is no \emph{cancellation}, in that
variables have to occur free \emph{at least once} when forming
abstractions. Terms of the $\lambda I$-calculus are not strongly
normalising in general. As an example, $\bm{\Omega}$ is a $\lambda I$-term.
One can prove, however, that on $\lambda I$-terms, weak normalisation
implies strong normalisation: in other words, all strategies are qualitatively
equivalent. In other words the equation $\Lambda_\mathsf{WN}=\Lambda_\mathsf{SN}$,
although not true in general holds in $\lambda I$. This does \emph{not} mean,
however, that all strategies are \emph{quantitatively} equivalent.

%%%%%%%%%%%%%%%%%%%%%%%%%%%%%%%%%%%%%
\subsubsection{The $\lambda A$-calculus.}
%%%%%%%%%%%%%%%%%%%%%%%%%%%%%%%%%%%%%
The $\lambda A$-calculus is the dual of $\lambda I$ and it is
often called \emph{affine} $\lambda$-calculus in the
literature. It is a very weak calculus in which variables bound by
abstractions occur \emph{at most once} free in the abstraction's body,
thus forbidding \emph{duplication}. The $\lambda A$-calculus is strongly
normalising, in the following strong sense: every reduction sequence from a term
$\termone$ has length bounded by the size of $\termone$. In other words,
the equation $\Lambda_\mathsf{WN}=\Lambda_\mathsf{SN}=\Lambda$,
although not true in general, holds in $\lambda A$.
%%%%%%%%%%%%%%%%%%%%%%%%%%%%%%%%%
\subsection{Reduction Strategies}
%%%%%%%%%%%%%%%%%%%%%%%%%%%%%%%%%
ARSs are sets endowed with a relation, and are thus a nondeterministic
model of computation. The notion of deterministic reduction strategy allows us
to fix one redex among the available ones, thus turning reduction into
a deterministic process.

\begin{definition}[Reduction Strategies]\label{def:strat}
	A \emph{reduction strategy} for an ARS $\mathcal{A}=(A,\red_\alpha)$ is a sub-ARS $(A,\red_\gamma)$ of $\mathcal{A}$, indicated with $\red_\gamma$ when the set of objects is clear from the context, having the same objects and normal forms of $\mathcal{A}$.
\end{definition}

\begin{definition}[Deterministic ARS]
	An ARS $(A,\rightarrow)$, is \emph{deterministic} if for each $a\in A$ there is at most one $b\in A$ such that $a\red b$.
	%$\mathsf{S}(a)\in\{b\,|\,a\rightarrow b\}$.
\end{definition}

\begin{definition}
	Given an ARS $(A,\rightarrow)$ and $a_0\in A$, a finite reduction sequence $\sigma:a_0\rightarrow a_1\rightarrow\cdots\rightarrow a_n$ is \emph{under strategy} $\red_\gamma$ if $a_i\red_\gamma a_{i+1}$ for each $0\leq i<n$. An infinite reduction sequence $\sigma:a_0\rightarrow a_1\rightarrow\cdots a_i\rightarrow\cdots$ is under strategy $\red_\gamma$ if $a_i\red_\gamma a_{i+1}$ for each $i\geq 0$.
\end{definition}

Since according to Definition~\ref{def:strat} reduction strategies are (sub-)ARSs, \emph{deterministic reduction strategies} are just reduction strategies which are deterministic ARSs. In the following, we will often employ a different definition of deterministic strategy, which turns out to be equivalent to the previous one, but more convenient.
\begin{definition}[Deterministic Strategies]\label{def:determstrat}
	Given an ARS $(A,\rightarrow)$, a deterministic reduction
	strategy for $A$ is a partial function $\mathsf{S}:A\rightharpoonup A$
	such that $\mathsf{S}(a)$ is defined iff $a$ is not in
	normal form and $a\rightarrow\mathsf{S}(a)$ whenever $\mathsf{S}(a)$
	is defined.
	%$\mathsf{S}(a)\in\{b\,|\,a\rightarrow b\}$.
\end{definition}

If $\sigma:a_0\rightarrow a_1\rightarrow\cdots\rightarrow a_n$ is a
reduction sequence under a deterministic strategy $\mathsf{S}$ and $a_n$ is in normal
form, we write $\nsteps{\mathsf{S}}(a_0)=n=|\sigma|$. If
$\sigma:a_0\rightarrow a_1\rightarrow\cdots$ is infinite, we say that
$\nsteps{\mathsf{S}}(a_0)=+\infty$. We can now define two deterministic 
reduction strategies for the $\lambda$-calculus that
will be useful in the following sections.
\begin{definition}
	\emph{Leftmost-outermost} $(\pslo)$ is a deterministic reduction
	strategy in which $\pslo(\termone)=\termtwo$ if and only if
	$\termone\redbeta\termtwo$ and the redex contracted in $\termone$ is
	the \emph{leftmost} among the ones in $\termone$ (measuring the
	position of a redex by its beginning symbol). If $\pslo(\termone)=\termtwo$,
	we write $\termone\redlo\termtwo$.
\end{definition}
Please note that the leftmost-outermost strategy is not always a deterministic strategy in the more general context of term rewriting systems \cite{terese_term_2003}, while in the $\lambda$-calculus it is. This will be fundamental in the definition of randomised strategies for the $\lambda$-calculus in next sections.
\begin{definition}
	\emph{Rightmost-innermost} $(\psri)$ is a deterministic reduction
	strategy in which $\psri(\termone)=\termtwo$ if and only if
	$\termone\redbeta\termtwo$ and the redex contracted in $\termone$ is
	the \emph{rightmost} among the ones in $\termone$ (measuring the
	position of a redex by its beginning symbol). Again, if
	$\psri(\termone)=\termtwo$, we write $\termone\redri\termtwo$.
\end{definition}

In the remainder of this section, we will give some basic preliminary
results about $\beta$-reduction, which form the basic blocks of our
subsequent analysis. Although standard, we prefer to state and prove
such results, for the sake of self-containedness. The following is a
technical lemma which basically says that $\beta$-reduction and
substitution commute. Crucially, the case in which the substituted
variable occurs \emph{at least} once is kept separate from the one in
which it occurs \emph{at most} once, this way allowing for a more
informative result.
\begin{lemma}\label{lemma:commutation}
	Let $\termone,\termtwo,\termthree\in\Lambda$. If $\termone\redbeta\termtwo$, then:
	\begin{varenumerate}
		\item $\termthree\{\termone/x\}\redbetasteps{\leq 1}\termthree\{\termtwo/x\}$ if $x$ is free at most once in $\termthree$;
		\item $\termthree\{\termone/x\}\redbetasteps{\geq 1}\termthree\{\termtwo/x\}$ if $x$ is free at least once in $\termthree$;
		\item $\termthree\{\termone/x\}\redbetared{}\termthree\{\termtwo/x\}$;
		\item $\termone\{\termthree/x\}\redbeta\termtwo\{\termthree/x\}$.
	\end{varenumerate}
\end{lemma}
\begin{proof}
	\mbox{}
	\begin{varenumerate}
		\item We proceed by induction on the structure of $\termthree$.
		\begin{varitemize}
			\item $\termthree=x$. Then $\termthree\{\termone/x\}=\termone\redbetasteps{1}\termtwo=\termthree\{\termtwo/x\}$.
			\item $\termthree=y$ and $y\neq x$. Then $\termthree\{\termone/x\}=y\redbetasteps{0}y=\termthree\{\termtwo/x\}$.
			\item $\termthree=\lambda y.\termfour$. Then $\termthree\{\termone/x\}=\lambda y.\termfour\{\termone/x\}$. By induction hypothesis $\termfour\{\termone/x\}\redbetasteps{\leq 1}\termfour\{\termtwo/x\}$ and thus $\termthree\{\termone/x\}=\lambda y.\termfour\{\termone/x\}\redbetasteps{\leq 1}\lambda y.\termfour\{\termtwo/x\}=\termthree\{\termtwo/x\}$.
			\item $\termthree=\termfour\termfive$. Then $\termthree\{\termone/x\}=\termfour\{\termone/x\}\termfive\{\termone/x\}$. If $x$ has no free occurrences in $\termthree$, then $\termthree\{\termone/x\}\redbetasteps{0}\termthree\{\termtwo/x\}$. Otherwise suppose the \emph{only} occurrence of $x$ is free in $\termfour$. Then $\termfive\{\termone/x\}=\termfive\{\termtwo/x\}$ and by induction hypothesis $\termfour\{\termone/x\}\redbetasteps{\leq 1}\termfour\{\termtwo/x\}$. Thus $\termthree\{\termone/x\}=\termfour\{\termone/x\}\termfive\{\termone/x\}\redbetasteps{\leq 1}\termfour\{\termtwo/x\}\termfive\{\termtwo/x\}=\termthree\{\termtwo/x\}$. The case in which $x$ is free in $\termfive$ is equivalent.
		\end{varitemize}
		\item The proof is equivalent to the one above in the variable and abstraction cases. One case is missing.
		\begin{varitemize}
			\item $\termthree=\termfour\termfive$. Then $\termthree\{\termone/x\}=\termfour\{\termone/x\}\termfive\{\termone/x\}$. If $x$ is free only in $\termfour$ or $\termfive$ the proof is the same of the above one. Otherwise let us suppose that $x$ occurs free both in $\termfour$ and in $\termfive$. Then we can apply twice the induction hypothesis which yields $\termfour\{\termone/x\}\redbetasteps{\geq 1}\termfour\{\termtwo/x\}$ and $\termfive\{\termone/x\}\redbetasteps{\geq 1}\termfive\{\termtwo/x\}$. Thus $\termthree\{\termone/x\}=\termfour\{\termone/x\}\termfive\{\termone/x\}\redbetasteps{\geq 1}\termfour\{\termtwo/x\}\termfive\{\termone/x\}\redbetasteps{\geq 1}\termfour\{\termtwo/x\}\termfive\{\termtwo/x\}=\termthree\{\termtwo/x\}$.
		\end{varitemize}
		\item The result directly follows by the previous two.
		\item We proceed by induction on the structure of $\termone$.
		\begin{varitemize}
			\item $\termone=x$. This case cannot occur because $\termone$ is a reducible term.
			\item $\termone=\lambda y.\termfour$. Then $\termone\{\termthree/x\}=\lambda y.\termfour\{\termthree/x\}$. Since $\termone\redbeta\termtwo$, then $\termtwo=\lambda y.\termfive$ and $\termfour\redbeta\termfive$. Thus by induction hypothesis $\termfour\{\termthree/x\}\redbeta\termfive\{\termthree/x\}$. Then $\termone\{\termthree/x\}=\lambda y.\termfour\{\termthree/x\}\redbeta\lambda y.\termfive\{\termthree/x\}=\termtwo\{\termthree/x\}$.
			\item $\termone=\termfour\termfive$ and the redex fired in $\termone$ to get $\termtwo$ is inside $\termfour$. Then $\termone\{\termthree/x\}=\termfour\{\termthree/x\}\termfive\{\termthree/x\}$, $\termtwo=\termsix\termfive$ and $\termfour\redbeta\termsix$. By induction hypothesis $\termfour\{\termthree/x\}\redbeta\termsix\{\termthree/x\}$. Thus $\termone\{\termthree/x\}=\termfour\{\termthree/x\}\termfive\{\termthree/x\}\redbeta\termsix\{\termthree/x\}\termfive\{\termthree/x\}=\termtwo\{\termthree/x\}$.
			\item $\termone=(\lambda y.\termfour)\termfive$ and $\termtwo=\termfour\{\termfive/y\}$. Then $\termone\{\termthree/x\}=(\lambda y.\termfour\{\termthree/x\})\termfive\{\termthree/x\}\redbeta\termfour\{\termthree/x\}\{\termfive\{\termthree/x\}/y\}$.  By Substitution Lemma this last term is equivalent to $\termfour\{\termfive/y\}\{\termthree/x\}=\termtwo\{\termthree/x\}$. Thus $\termone\{\termthree/x\}\redbeta\termtwo\{\termthree/x\}$.\qed
		\end{varitemize}
	\end{varenumerate}
\end{proof}
Next, we need to give some quantitative refinements of the usual
confluence result, which are however possible only in
\emph{fragments} of the $\lambda$-calculus rather than in the $\lambda$-calculus
itself. Let us first consider the case of the $\lambda A$-calculus. We can observe in this
case that the leftmost-outermost strategy has some peculiarity over the
other ones, in that every peak $\termtwo\coredlo\termone\redbeta\termthree$
can be \emph{closed} by a valley in which the path from $\termtwo$ \emph{is not
  longer} than the one from $\termthree$, hinting at the efficiency of
$\pslo$.
\begin{lemma}\label{lemma:diamond}
	Let $\termone$ be a $\lambda A$-term. If $\termone\redlo\termtwo$ and $\termone\redbeta\termthree$, then there exists a term $\termfour$ such that $\termtwo\redbetasteps{n}\termfour$, $\termthree\redlosteps{m}\termfour$ and $n\leq m\leq 1$.
\end{lemma}
\begin{proof}
	Let us call $\rdxone=(\lambda x.\termfive)\termsix$ the $\pslo$-redex of 
	$\termone$ and $\rdxtwo$ the redex fired in the reduction from $\termone$ 
	to $\termthree$. Let $\termone=\contone[\rdxone]$. We proceed by induction 
	on the structure of $\contone$.
	\begin{varitemize}
		\item $\contone=\Box$. Then $\termone=\rdxone$ and $\termtwo=\termfive\{\termsix/x\}$. We distinguish some cases.
		\begin{varitemize}
			\item $\rdxone$ is $\rdxtwo$. Then $\termtwo=\termthree=\termfour$.
			\item $\rdxtwo$ is inside $\termfive$. Call $\termfive'$ the contractum of $\termfive$ obtained firing redex $\rdxtwo$. Then $\termthree=(\lambda x.\termfive')\termsix\redlo\termfive'\{\termsix/x\}$. By Lemma~\ref{lemma:commutation}.4, $\termtwo=\termfive\{\termsix/x\}\redbeta\termfive'\{\termsix/x\}$.
			\item $\rdxtwo$ is inside $\termsix$. Call $\termsix'$ the contractum of $\termsix$ obtained firing redex $\rdxtwo$. Then $\termthree=(\lambda x.\termfive)\termsix'\redlo\termfive\{\termsix'/x\}$. By Lemma~\ref{lemma:commutation}.1, $\termtwo=\termfive\{\termsix/x\}\redbetasteps{\leq 1}\termfive\{\termsix'/x\}$.
		\end{varitemize}
		\item $\contone=\termseven\conttwo$. $\rdxtwo$ cannot be inside $\termseven$, otherwise $\rdxone$ would not be the $\pslo$-redex. Thus the result follows by induction hypothesis.
		\item $\contone=\conttwo\termseven$. If $\rdxtwo$ is inside 
		$\conttwo[\rdxone]$, then the result follows by induction hypothesis. 
		Otherwise $\rdxone$ and $\rdxtwo$ are independent. Then 
		$\termone=\conttwo[\rdxone]\contthree[\rdxtwo]$. Moreover, 
		$\termtwo=\conttwo[\rdxone']\contthree[\rdxtwo]$ and 
		$\termthree=\conttwo[\rdxone]\contthree[\rdxtwo']$, where $\rdxone'$ 
		and $\rdxtwo'$ are the contracta of $\rdxone$ and $\rdxtwo$, 
		respectively. Thus, 
		$\termtwo=\conttwo[\rdxone']\contthree[\rdxtwo]\redbetasteps{1}\conttwo[\rdxone']\contthree[\rdxtwo']=\termfour$
		 and 
		$\termthree=\conttwo[\rdxone]\contthree[\rdxtwo']\redlosteps{1}\conttwo[\rdxone']\contthree[\rdxtwo']=\termfour$.
		\item $\contone=\lambda x.\conttwo$. The result follows by induction hypothesis.\qed
	\end{varitemize} 
\end{proof}
The situation in $\lambda I$ is dual: peaks can be closed (possibly) faster
by \emph{not} going leftmost outermost.
\begin{lemma}\label{lemma:diamond2}
	Let $\termone$ be a $\lambda I$-term. If $\termone\redlo\termtwo$ and $\termone\redbeta\termthree$, then there exists a term $\termfour$ such that $\termtwo\redbetasteps{n}\termfour$, $\termthree\redlosteps{m}\termfour$ and $n\geq m$.
\end{lemma}
\begin{proof}
	The proof is equivalent to the one of Lemma~\ref{lemma:diamond}, using Lemma~\ref{lemma:commutation}.2 instead of Lemma~\ref{lemma:commutation}.1.\qed
\end{proof}
Thirdly, we can consider the case of the full $\lambda$-calculus. Here, the only
thing which can be said is that peaks can be solved in such a way as to guarantee
that the valley has at most unitary length \emph{on the $\pslo$ side}.
\begin{lemma}\label{lemma:diamond3}
	Let $\termone$ be a $\lambda$-term. If $\termone\redlo\termtwo$ and 
	$\termone\redbeta\termthree$, then there exists a term $\termfour$ such 
	that $\termtwo\redbetared{}\termfour$, $\termthree\redlosteps{m}\termfour$ 
	and $m\leq 1$. In particular $m=0$ if and only if $\termtwo=\termthree$.
\end{lemma}
\begin{proof}
	The proof is equivalent to the one of Lemma~\ref{lemma:diamond}, using Lemma~\ref{lemma:commutation}.3 instead of Lemma~\ref{lemma:commutation}.1.\qed
\end{proof}
Finally, we can prove that $\nsteps{\pslo}(\cdot)$, seen as a function on $\lambda$-terms,
can never increase along \emph{any} $\beta$-reduction.
\begin{lemma}\label{lemma:derivationlength}
	If $\termone\redbeta\termtwo$, then $\nsteps{\pslo}(\termtwo) \leq \nsteps{\pslo}(\termone)$.
\end{lemma}
\begin{proof}
	We argue by induction on $\nsteps{\pslo}(\termone)$. We call $\rdxone$ the 
	redex contracted from $\termone$ to $\termtwo$, $\termthree$ the 
	$\pslo$-reduct of $\termone$ and $\rdxtwo$ the $\pslo$-redex of $\termone$. 
	The base of the induction is $\nsteps{\pslo}(\termone)=1$ and it is proved 
	as follows. By Lemma~\ref{lemma:diamond3}, there exists a term $\termfour$ 
	such that $\termthree\redbetared{}\termfour$, 
	$\termtwo\redlosteps{m}\termfour$ and $m\leq 1$. Since 
	$\nsteps{\pslo}(\termone)=1$, $\termthree$ is in normal form and thus 
	coincides with $\termfour$. Thus $\nsteps{\pslo}(\termtwo)\leq 1$. Now 
	suppose this Lemma true for each term $\termfour$ such that 
	$\nsteps{\pslo}(\termfour)\leq k$. Let us consider a term $\termone$ with 
	$\nsteps{\pslo}(\termone)=k+1<\infty$. Call $\termthree$ the $\pslo$-reduct 
	of 
	$\termone$. By Lemma~\ref{lemma:diamond3}, there exists a term $\termfour$ 
	such that $\termtwo\redlosteps{m}\termfour$ and 
	$\termthree\redbetared{}\termfour$. The case in which $\termtwo$ and 
	$\termthree$ coincide and thus $m=0$ is trivial, so let us consider 
	distinct $\termtwo$ and $\termthree$, which implies $m=1$. 
	$\nsteps{\pslo}(\termthree)=k$. Thus, by induction hypothesis, 
	$\nsteps{\pslo}(\termfour)\leq k$. 
	$\nsteps{\pslo}(\termtwo)=1+\nsteps{\pslo}(\termfour)\leq 
	1+k=\nsteps{\pslo}(\termone)$. If $\nsteps{\pslo}(\termone)=\infty$, then 
	$\nsteps{\pslo}(\termtwo) \leq \nsteps{\pslo}(\termone)$ since 
	$\nsteps{\pslo}(\termtwo)\leq\infty$. \qed
\end{proof}
\subsection{Minimal and Maximal Strategies in sub-$\lambda$-calculi}
Since we exploit ARSs as models of computation, we are interested in defining 
quantitative properties on them. If one is interested in efficiency, it makes a 
lot of sense, of course, to implement an ARS with a strategy that minimises the 
number of steps to normal form. On the other hand, one could be interested in a 
worst-case analysis, thus seeking a maximal strategy.
\begin{definition}[\cite{vanoostrom_2007}]
	Let $(A,\red)$ be an ARS and $\red_\gamma$ a strategy for it.
	\begin{varitemize}
		\item  $\red_\gamma$ is \emph{normalising}, if every weakly normalising object only allows finite maximal reduction sequences under $\red_\gamma$.
		\item $\red_\gamma$ is \emph{minimal}, if the length of any reduction sequence under $\red_\gamma$ from any object to normal form is minimal among the lengths of all the reduction sequences from the former to the latter.
		\item $\red_\gamma$ is \emph{perpetual}, if every object which is not strongly normalising only allows infinite maximal reduction sequences under $\red_\gamma$.
		\item $\red_\gamma$ is \emph{maximal}, if the length of any reduction sequence under $\red_\gamma$ from any object to normal form is maximal among the lengths of all the reduction sequences from the former to the latter.
	\end{varitemize}
\end{definition}
It is well-known that the minimal and normalising strategy is \emph{not} computable in the scope of the full $\lambda$-calculus \cite[Section~13.5]{barendregt_lambda_1984}, i.e. one could not select the redex leading to the minimal reduction sequence to normal form in an effective way. Instead, there exists an effective maximal and perpetual strategy \cite{vanraamsdonk_1999}. In this Section we prove that in $\lambda A$ (respectively, in $\lambda I$) the leftmost-outermost strategy is minimal and normalising (respectively, maximal and perpetual). In order to prove this, we exploit a result due to van Oostrom~\cite{vanoostrom_2007}.
\begin{definition}
	Let $(A,\red)$ be an ARS and $\red_\alpha,\red_\gamma$ two reduction strategies for it. Suppose that for every $a\in A$, if $b\leftarrow_\alpha a\red_\gamma c$, then either there is an infinite reduction sequence from $c$ under $\red_\alpha$ or there is $d\in A$ such that $b\red_\gamma^n d\leftarrow_\alpha^m c$ and $n\leq m$. We say that $\red_\alpha$ \emph{ordered locally commutes with $\red_\gamma$}, abbreviated $\mathbf{OLCOM}(\red_\alpha,\red_\gamma)$.
\end{definition}
\begin{theorem}[\cite{vanoostrom_2007}]\label{theorem:vanoostrom}
	Let $(A,\red)$ be an ARS and $\red_\alpha$ a reduction strategy for it.
	\begin{varenumerate}
		\item If $\textbf{OLCOM}(\red_\alpha,\red)$, then $\red_\alpha$ is minimal and normalising.
		\item If $\textbf{OLCOM}(\red,\red_\alpha)$, then $\red_\alpha$ is maximal and perpetual.
	\end{varenumerate}
\end{theorem}

This way, we have a proof method that allow us to state and prove in a simple way that the $\pslo$ strategy is minimal and normalising (respectively, maximal and perpetual) in $\lambda A$ (respectively, in $\lambda I$).

\begin{theorem}\label{theorem:optpes}
	In $\lambda A$ $\pslo$ is a minimal and normalising strategy, in $\lambda I$ it is a maximal and perpetual one.
\end{theorem}
\begin{proof}
	In $\lambda A$, by Lemma~\ref{lemma:diamond}, $\textbf{OLCOM}(\mathsf{\pslo},\redbeta)$ holds and in $\lambda I$, by Lemma~\ref{lemma:diamond2}, $\textbf{OLCOM}(\redbeta,\mathsf{\pslo})$ holds. Then the result follows by Theorem~\ref{theorem:vanoostrom}.\qed
\end{proof}
One might guess that a dual result holds, i.e. the rightmost-innermost strategy is minimal for $\lambda I$ and maximal for $\lambda A$. However this is not the case, as witnessed by the two following counterexamples (Figure~\ref{figure:counterexample}). The former from \cite{asperti_1998}, the latter provided by Damiano Mazza in a personal communication.
\begin{varitemize}
	\item In $\lambda I$, $\psri$ is not minimal: $\termone_I=(\lambda x.x\bm{I})(\lambda x.(\lambda z.zz)(xy))$ reduces rightmost to normal form in five steps, while there exists a reduction sequence to normal form of four steps.
	\item In $\lambda A$, $\psri$ is not maximal: $\termone_A=(\lambda x.x\bm{I})(\lambda x.(\lambda z.y)(xy))$ reduces rightmost to normal form in three steps, while there exists a reduction sequence to normal form of four steps.
\end{varitemize}
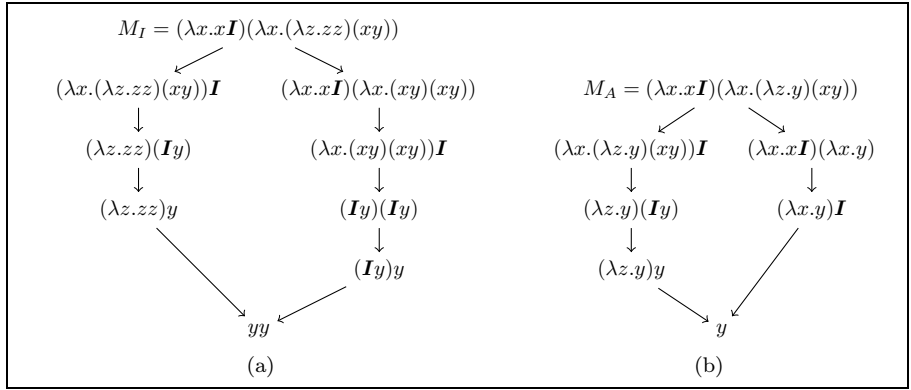
\begin{figure}
	\fbox{
		\begin{minipage}{.96\textwidth}
			\centering
			\subfloat[]{{
				\begin{tikzpicture}
				[node distance=20mm, auto, transform shape,scale=0.8]
				\node (m) at (0,0) {$\termone_I=(\lambda x.x\bm{I})(\lambda x.(\lambda z.zz)(xy))$};
				\node (n) at (-2,-1) {$(\lambda x.(\lambda z.zz)(xy))\bm{I}$};
				\node (l) at (-2,-2) {$(\lambda z.zz)(\bm{I}y)$};
				\node (q) at (-2,-3) {$(\lambda z.zz)y$};
				\node (w) at (2,-1) {$(\lambda x.x\bm{I})(\lambda x.(xy)(xy))$};
				\node (e) at (2,-2) {$(\lambda x.(xy)(xy))\bm{I}$};
				\node (r) at (2,-3) {$(\bm{I}y)(\bm{I}y)$};
				\node (t) at (2,-4) {$(\bm{I}y)y$};
				\node (u) at (0,-5) {$yy$};
				\draw (m) edge[->] node[left=5pt] {} (n);
				\draw (n) edge[->] node[left=5pt] {} (l);
				\draw (l) edge[->] node[left=5pt] {} (q);
				\draw (q) edge[->] node[left=5pt] {} (u);
				\draw (m) edge[->] node[left=5pt] {} (w);
				\draw (w) edge[->] node[left=5pt] {} (e);
				\draw (e) edge[->] node[left=5pt] {} (r);
				\draw (r) edge[->] node[left=5pt] {} (t);
				\draw (t) edge[->] node[right=5pt] {} (u);
				\end{tikzpicture}}}
			\qquad
			\subfloat[]{{
					\begin{tikzpicture}
					[node distance=20mm, auto, transform shape,scale=0.8]
					\node (p) at (0,1) {$\termone_A=(\lambda x.x\bm{I})(\lambda x.(\lambda z.y)(xy))$};
					\node (m) at (-1.5,0) {$(\lambda x.(\lambda z.y)(xy))\bm{I}$};
					\node (n) at (-1.5,-1) {$(\lambda z.y)(\bm{I}y)$};
					\node (l) at (-1.5,-2) {$(\lambda z.y)y$};
					\node (w) at (1.5,-1) {$(\lambda x.y)\bm{I}$};
					\node (y) at (1.5,0) {$(\lambda x.x\bm{I})(\lambda x.y)$};
					\node (u) at (0,-3) {$y$};
					\draw (m) edge[->] node[left=5pt] {} (n);
					\draw (p) edge[->] node[left=5pt] {} (m);
					\draw (n) edge[->] node[left=5pt] {} (l);
					\draw (l) edge[->] node[left=5pt] {} (u);
					\draw (p) edge[->] node[left=5pt] {} (y);
					\draw (y) edge[->] node[left=5pt] {} (w);
					\draw (w) edge[->] node[left=5pt] {} (u);
					\end{tikzpicture}}}
		\end{minipage}}
		\caption{The right-hand side of each diagram represents the reduction sequence under rightmost-innermost strategy. On the left-hand side a shorter (a) and longer (b) reduction sequences are provided.}
		\label{figure:counterexample}
	\end{figure}
Please observe that the two counterexamples are built behind very simple principles. In both of them there is the \emph{virtual} redex $xy$ that is copied (respectively, cancelled) before being contracted, if one reduces rightmost-innermost.
%%%%%%%%%%%%%%%%%%%%%%%%%%%%%%%%%%%%%%%%%%%%%%%%%%%%%%%%%%%%%%%%%%%%%%%%%%%%%%%%%%%%%%%%%%%%%%%%%%%%%%%%%%%%%%%%%%%%%%%%%%%%%%%%
\section{Probabilistic Abstract Reduction Systems as Strategies}\label{section:prob}
%%%%%%%%%%%%%%%%%%%%%%%%%%%%%%%%%%%%%%%%%%%%%%%%%%%%%%%%%%%%%%%%%%%%%%%%%%%%%%%%%%%%%%%%%%%%%%%%%%%%%%%%%%%%%%%%%%%%%%%%%%%%%%%%
We introduce now a framework suitable to define randomised
strategies. In particular, we shift the notion of ARS to the fully
probabilistic case, where nondeterminism is solved through probabilistic choice, obtaining this way a model akin to Markov chains. The first preliminary concept we need is that of a distribution.
\begin{definition}[Distribution]
  A \emph{partial probability distribution} over a countable set $A$
  is a mapping $\rho:A\rightarrow\left[0,1\right]$ such that
  $\left|\rho\right|\leq 1$ where
  $\left|\rho\right|=\underset{a\in A}{\sum}\rho\left(a\right)$. We denote the set of partial probability distributions over $A$
  by $\pdist{A}$. The \emph{support} of a partial distribution
  $\rho\in\pdist{A}$ is the set $\supp{\rho}=\left\lbrace a \in A \, |
  \, \rho\left( a\right) >0\right\rbrace $.  A \emph{probability
    distribution} over a countable set $A$ is a partial probability
  distribution $\mu$ such that $\left|\mu\right|=1$. $\dist{A}$ denotes the set
  of probability distributions over $A$.
\end{definition}
Strategies as from Definition~\ref{def:determstrat} are
deterministic: the process of picking a reduct among the many possible
ones can only have \emph{one} outcome. But what if this process becomes
\emph{probabilistic}? This is captured by the following notions.
\begin{definition}[Randomised Strategies]
  Given an ARS $(A,\rightarrow)$, a
  \emph{randomised reduction strategy $\mathsf{P}$ for $(A,\rightarrow)$}
  is a partial function $\mathsf{P}:A\rightharpoonup\dist{A}$ such that $\mathsf{P}(a)$ is undefined if and only if $a\in A$ is in normal form. Moreover, if $\mathsf{P}(a)=\mu$, then
  $\supp{\mu}\subseteq\{b\,|\,a\rightarrow b\}$.
  %\emph{fully probabilistic abstract
  %  reduction system (FPARS) for $(S,\rightarrow)$} is a pair $\left(
  %S,\mathsf{P}\right)$ where $\mathsf{P}:S\rightharpoonup\dist{S}$ is
  % We call $\mathsf{P}$
  %a \emph{randomised reduction strategy for $(S,\rightarrow)$}.
\end{definition}

Intuitively, the last constraint ensures that one can pass with positive 
probability from an object $a$ to an object $b$ only if $a\red b$.

\begin{definition}[FPARS]
	If $(A,\rightarrow)$ is an ARS and 
	$\mathsf{P}$ is a randomised reduction strategy for it,
	then we call $(A,\mathsf{P})$ a \emph{fully} probabilistic abstract
	reduction system (FPARS).
\end{definition}
Our model is said to be \emph{fully} probabilistic in that nondeterminism is 
not taken into account. Indeed PARSs as defined in 
\cite{bournez_proving_2005,avanzini_probabilistic_2018} combine a 
nondeterministic behaviour with a randomised one. FPARSs dynamics instead is 
purely probabilistic. In the following, we will
study randomised strategies seen as FPARSs.

The dynamics of an FPARS can be handled by way of an appropriate
notion of a configuration, on which an evolution function can be
defined.
\begin{definition}[Configurations, Computations]
Let $\left( A,\mathsf{P}\right)$ be an FPARS and $a,b\in A$ be two
\emph{states}. We define the \emph{probability $\mathbb{P}\left(a
\rightarrow b\right)$ of a transition} from $a$ to $b$:
$$
\mathbb{P}\left(a \rightarrow b\right)=
\begin{cases}
\mu\left(b\right)&\mbox{if $\mathsf{P}\left(a\right)=\mu$},\\
0&\mbox{if $\mathsf{P}\left(a\right)$ is undefined}.
\end{cases}
$$
A \emph{configuration} of an FPARS $\left( A,\mathsf{P}\right)$ is a partial probability distribution $\rho\in\pdist{A}$.
The \emph{evolution} of an FPARS $\left( A,\mathsf{P}\right)$ is the function 
$\mathsf{E}:\pdist{A}{\rightarrow\pdist{A}}$ defined as follows:
$$
\mathsf{E}\left( \rho\right) = \sigma \textnormal{ where } \sigma\left(a\right)=\underset{b\in A}{\sum}\rho\left( b\right) \cdot\mathbb{P}\left(b \rightarrow a\right)\;\textnormal{for every }a \in A.
$$
If  $\mathsf{E}\left( \rho\right) = \sigma$ we write $\rho\rightsquigarrow\sigma$.
A \emph{computation} is any sequence
$(\rho_i)_{i\in\mathbb{N}}$, such that $\rho_i\rightsquigarrow\rho_{i+1}$.
\end{definition}
\begin{remark}
  Those computations $(\rho_i)_{i\in\mathbb{N}}$ where
  $\rho_0$ is Dirac (i.e. there exists $a\in A$ such that
  $\rho_0 (a)=1$) are particularly interesting: they model
  the evolution of an FPARS starting from a state $a \in A$. We write in this 
  case
  $\rho_0=\mathbf{Dirac}(a)$. In the following, we will always consider 
  computations of this type.
\end{remark}

\begin{remark}
	Please note that if $\rho\rightsquigarrow\sigma$, then $|\sigma|\leq|\rho|$. As proposed in \cite{avanzini_probabilistic_2018}, the total probability mass decreases along a reduction sequence, becoming $0$ when the computation has terminated. In particular, if $\rho=\mathbf{Dirac}(a)$ and $a$ is in normal form, then $\mathsf{E}(\rho)=\sigma$, where $|\sigma|=0$.
\end{remark}

\begin{example}\label{example:pars}
	Consider an ARS $(A,\rightarrow)$, where $A=\{a,b\}$ and 
	$\rightarrow\,=\{(a,a),(a,b)\}$. We define a randomised strategy 
	$\mathsf{P}$ on top of $(A,\rightarrow)$. $\mathsf{P}(a)=\mu$ where 
	$\mu(a)=\mu(b)=\frac{1}{2}$, while $\mathsf{P}(b)$ is undefined (since $b$ 
	is in normal form). A computation $(\rho_i)_{i\in\mathbb{N}}$ starting from 
	$\rho_0=\mathbf{Dirac}(a)$ has the following form.
	$$
	\underset{\rho_{0}}{\begin{cases}
		1 & a\\
		0 & b
		\end{cases}}\rightsquigarrow\underset{\rho_{1}}{\begin{cases}
		\frac{1}{2} & a\\
		\frac{1}{2} & b
		\end{cases}}\rightsquigarrow\underset{\rho_{2}}{\begin{cases}
		\frac{1}{4} & a\\
		\frac{1}{4} & b
		\end{cases}}\rightsquigarrow\cdots\rightsquigarrow\underset{\rho_{k}}{\begin{cases}
		\frac{1}{2^{k}} & a\\
		\frac{1}{2^{k}} & b
		\end{cases}}\rightsquigarrow\cdots
	$$
\end{example}
How could we measure the \emph{length} of a computation? It is natural
to look for a definition capturing the \emph{average} derivation length from
$s$ to its normal form. 
\begin{definition}
  Let $(A,\mathsf{P})$ be an FPARS and $\rho_0=\mathbf{Dirac}(a)$,
  where $a\in A$. Given the computation $(\rho_i)_{i\in\mathbb{N}}$,
  $\nsteps{\mathsf{P}}(a)=\sum\limits_{i=1}^{\infty} \left|\rho_i\right|$.
\end{definition}
The definition above collapses to the one given for the
deterministic case when $\mathsf{P}$ is deterministic. Moreover, observe that 
$\nsteps{\mathsf{P}}:A\rightarrow\mathbb{R}\cup\{+\infty\}$, i.e. the average 
length of a computation could be infinite. A witness of such behaviour is the 
FPARS $(A,\mathsf{P})$ where $A=a$ and $\mathsf{P}(a)=\mu$ such that 
$\mu(a)=1$. In the next Section we provide a justification for this definition 
considering the length of a computation as an appropriately defined random 
variable. Besides quantitative information on the length, one might just be 
interested in knowing if a reduction is finite or infinite. Indeed, termination 
is a crucial problem in rewriting theory. Since we are in a probabilistic 
context, distinct such notions are
possible. We define in our setting two classical termination
properties, namely almost-sure termination and positive almost-sure termination \cite{bournez_proving_2005,ferrer_fioriti_probabilistic_2015,avanzini_probabilistic_2018}.
\begin{definition}
	An FPARS is \emph{almost-surely terminating} (AST) if each computation $(\rho_i)_{i\in\mathbb{N}}$, is such that $\underset{n\rightarrow +\infty}{\lim}\left|\rho_{n}\right|=0$.
\end{definition}
\begin{definition}
	An FPARS $(A,\mathsf{P})$ is \emph{positive almost-surely terminating} (PAST) if for each $a \in A$,
	$\nsteps{\mathsf{P}}(a)< +\infty$.
	In this case we say that $\mathsf{P}$ is a \emph{positive almost-surely normalising} strategy.
\end{definition}
\begin{example}
	Consider the same setting of Example \ref{example:pars}. It is easy to see that $(A,\mathsf{P})$ is AST since $\underset{n\rightarrow +\infty}{\lim}\left|\rho_{n}\right|=\underset{n\rightarrow +\infty}{\lim}\frac{1}{2^{n-1}}=0$. Moreover $(A,\mathsf{P})$ is PAST since $\nsteps{\mathsf{P}}(a)=\sum\limits_{n=1}^{\infty} \left|\rho_n\right|=\sum\limits_{n=1}^{\infty}\frac{1}{2^{n-1}}=2$ and $\nsteps{\mathsf{P}}(b)=0$.
\end{example}
While PAST clearly implies AST, it is well-known from Markov chain literature that AST does not imply PAST e.g. in the symmetric random walk on $\mathbb{Z}$ \cite[Chapter 1.7]{norris_markov_1998}. We prove in our framework a result analogous to a classical Theorem in Markov chain theory due to Foster \cite{foster1953} that gives a sufficient condition for PAST and a bound on the average number of steps to normal form.
\begin{notation}
	For $\varepsilon>0$ we write $x>_{\varepsilon}y$ if and only if $x\geq y+\varepsilon$. This order is well-founded on real numbers with a lower bound. Please note that if $\varepsilon=0$, then the order is not well founded on the reals with a lower bound. For example, one could have an infinite descending chain of strictly positive reals like $1>\frac{1}{2}>\frac{1}{4}>\cdots>\frac{1}{2^k}>\cdots$.
\end{notation}
\begin{definition}
	Given an FPARS $\left( A,\mathsf{P}\right)$, we define a function $V:A\rightarrow\mathbb{R}$ as Lyapunov \cite{bremaud_markov_1999} if the following are satisfied.
	\begin{varenumerate}
		\item
		There exists $k\in\mathbb{R}$ such that $V\left(a\right)\geq k$ for each $a\in A$.
		\item
		There exists $\varepsilon>0$ such that for every $a\in A$ if $\mathsf{P}\left(a\right) = \mu$, then $V\left(a\right) >_{\varepsilon} V\left(\mu\right)$, where $V$ is extended to partial distributions as follows:
		$$
		V\left(\mu\right)=\sum_{b\in A}V\left(b\right)\cdot\mu\left(b\right).
		$$
	\end{varenumerate}
\end{definition}
\begin{remark}
	Without loss of generality, given a Lyapunov function $V$ we can always consider a new Lyapunov function $W(a)\geq 0$ for each $a\in A$ simply adding a constant to $V$. Thus in the following we will always assume Lyapunov functions to be non-negative.
\end{remark}
\begin{theorem}
	If we can define for an FPARS $\mathcal{P}=\left(A,\mathsf{P}\right)$ a Lyapunov function $V$, then $\mathcal{P}$ is PAST and the average derivation length $\nsteps{\mathsf{P}}(a)$ of any sequence $(\rho_i)_{i\in\mathbb{N}}$ starting from any $a\in A$ is bounded by
	$\frac{V\left(a\right)}{\varepsilon}$.
\end{theorem}
	\begin{proof}
		Let us consider a generic transition $\rho_{i-1}\rightsquigarrow\rho_i$ of $\mathcal{P}$.
		\[
		0\leq V(\rho_i)=\sum_{k\in A}V(k)\rho_i(k)=\sum_{k\in A}V(k)\sum_{j\in A}\rho_{i-1}(j)\cdot\mathbb{P}(j\rightarrow k)
		\]
		If $j\in\mathbf{NF}(A)$, then $\mathbb{P}(j\rightarrow k)=0$, otherwise 
		$\mathbb{P}(j\rightarrow k)=\mu_j(k)$, where $\mu_j=\mathsf{P}(j)$. 
		Thus, resuming the chain of inequalities,
		\begin{equation*}
		\begin{split}
		0&\leq\sum_{k\in A}V(k)\sum_{j\in A\setminus 
		\mathbf{NF}(A)}\rho_{i-1}(j)\cdot\mu_j(k)=\sum_{j\in A\setminus 
		\mathbf{NF}(A)}\rho_{i-1}(j)\sum_{k\in A}V(k)\mu_j(k)\\
		&=\sum_{j\in A\setminus \mathbf{NF}(A)}\rho_{i-1}(j)\cdot V(\mu_j)\leq 
		\sum_{j\in A\setminus \mathbf{NF}(A)}\rho_{i-1}(j)\cdot 
		(V(j)-\varepsilon)\\
		&\leq \sum_{j\in A}\rho_{i-1}(j)\cdot V(j) - \sum_{j\in A\setminus \mathbf{NF}(A)}\rho_{i-1}(j)\cdot \varepsilon=V(\rho_{i-1})-\varepsilon\cdot|\rho_i|.
		\end{split}
		\end{equation*}
		Iterating the above inequality we have
		\begin{equation*}
		\begin{split}
		0&\leq V(\rho_i)\leq V(\rho_{i-1})-\varepsilon\cdot|\rho_i|\\
		&\leq V(\rho_{i-2})-\varepsilon\cdot(|\rho_i| + |\rho_{i-1}|)\leq \cdots\leq V(\rho_0)-\varepsilon\sum_{n=1}^i |\rho_n|.
		\end{split}
		\end{equation*}
		Taking the limit for $i\rightarrow +\infty$ this yields to $0\leq V(\rho_0)-\varepsilon\cdot\nsteps{\mathsf{P}}(a)$ and thus
		$$
		\nsteps{\mathsf{P}}(a)\leq \frac{V(\rho_0)}{\varepsilon}=\frac{V(a)}{\varepsilon}.
		$$
		\qed
	\end{proof}
\subsection{FPARSs and Markov Chains}
All the machinery introduced up to now avoids the use of formal probability 
theory, namely probability spaces, measures and random variables. This was done 
on purpose, in order to provide a simple mathematical framework in which our 
theory could be settled. However, it seems appropriate to relate our 
definitions to those which appear in the literature about stochastic processes 
\cite{ash_probability_1999,bremaud_markov_1999}. We first recall some basic 
notions of probability theory and Markov chain theory. Please refer to the 
references for a more detailed discussion on the subject.
\begin{definition}
	These definitions are standard in probability theory and are adapted from \cite{ash_probability_1999,ferrer_fioriti_probabilistic_2015}.
	\begin{varitemize}
		\item A \emph{probability space} is a triple $(\Omega, \mathcal{F},\mathbb{P})$ where $\Omega$ is the \emph{sample space}, $\mathcal{F}\subseteq2^\Omega$ is a $\sigma$-algebra on $\Omega$ whose elements are said to be \emph{events}, and $\mathbb{P}$ is a \emph{probability measure} for $\mathcal{F}$ i.e. $\mathbb{P} :\mathcal{F}\rightarrow[0,1]$, $\mathbb{P}\{\Omega\}=1$, and $\mathbb{P}$ is countably additive: the probability of a countable disjoint union of events is the sum of the individual probabilities.
		\item Given two events $A$ and $B$ if $\mathbb{P}\{B\}\neq 0$ the \emph{conditional probability} of $A$ given $B$ is defined as 
		$$
		\mathbb{P}\{A|B\}=\frac{\mathbb{P}\{A\cap B\}}{\mathbb{P}\{B\}}.
		$$
		\item Let $(\Omega, \mathcal{F},\mathbb{P})$ be a probability space. A \emph{random variable} $X$ is a function from $\Omega$ to $\mathbb{R}\cup\{\pm\infty\}$ such that for each $x\in \mathbb{R}\cup\{\pm\infty\}$, the set $\{X\leq x\}$:=$\{\omega\in\Omega : X(\omega)\leq x\}\in \mathcal{F}$. A random variable $X$ is \emph{discrete} if the range of $X$ is finite or countably infinite.
		\item Given a random variable $X$ on $(\Omega, \mathcal{F},\mathbb{P})$, its \emph{expected value} is 
		$$
		\mathbb{E}[X]=\int_{\Omega}X\mathrm{d}\mathbb{P},
		$$
		where $\int$ is the Lebesgue integral. In the case $X$ is a discrete random variable assuming values in $I$,
		$$
		\mathbb{E}[X]=\sum_{i\in I}i\cdot\mathbb{P}\{X=i\}.
		$$
	\end{varitemize}
\end{definition}
We give an alternative way of computing the expected value of a discrete random variable that will be useful in the remainder of this section.
\begin{proposition}[Telescope Formula, \cite{bremaud_markov_1999}]\label{prop:alt-def}
	Let $X$ be a discrete random variable with values in $\mathbb{N}\cup \{+\infty\}$. Then
	$$
	\mathbb{E}[X]=\sum\limits_{i=1}^{\infty} \mathbb{P}\{X\geq i\}.
	$$
\end{proposition}
\begin{proof}
	Since $X$ is discrete 
	$$
	\mathbb{E}[X]=\sum\limits_{k=1}^{\infty} k\cdot\mathbb{P}\{X=k\} =  \sum\limits_{k=1}^{\infty} \sum\limits_{i=1}^{k}\mathbb{P}\{X=k\}.
	$$
	We can exchange the summations by Fubini-Tonelli Theorem that yields 
	$$
	\mathbb{E}[X]=\sum\limits_{i=1}^{\infty} \sum\limits_{k=i}^{\infty} \mathbb{P}\{X=k\} = \sum\limits_{i=1}^{\infty} \mathbb{P}\{X\geq i\}.
	$$\qed
\end{proof}
If random variables represent the state of a dynamical system through time, we 
call this system a \emph{stochastic process}. In this work, the reduction to 
normal form of a $\lambda$-term is a stochastic process, and in particular a 
discrete time Markov chain. These are a well-studied class of memoryless 
stochastic processes. In fact, they satisfy the Markov property, which means 
that transitions from the current state to the next one do not depend on the 
whole history of the process, but only on the current state.
\begin{definition}[Markov Chain]
	Given a probability space $(\Omega, \mathcal{F},\mathbb{P})$, and a state space $I$ at most countable and measurable with the $\sigma$-algebra of all its subsets, a sequence of random variables $(X_n)_{n\geq 0}$ with values in $I$ is a\emph{ Markov Chain} if and only if it satisfies the \emph{Markov property}, i.e. for each $n\geq 0$, for each $i_0,..,i_n,j \in I$, such that $\mathbb{P}\{X_n=i_n,...,X_0=i_0\}>0$ it holds that
	$$
	\mathbb{P}\{X_{n+1}=j|X_n=i_n,...,X_0=i_0\}=\mathbb{P}\{X_{n+1}=j|X_n=i_n\}=p_{{i_n}j}.
	$$
\end{definition}
We can completely characterise a Markov chain $\mathcal{M}$ as a tuple 
$(I,X_0,P)$ where $I$ is the \emph{state space}, $X_0$ is the \emph{initial 
distribution} and $P=(p_{ij})_{,i,j\in I}$ is the \emph{transition matrix}. The 
first entrance time into a state is a very useful random variable in order to 
study the behaviour of a Markov chain. Particularly interesting is knowing 
whether it is almost-surely finite and with finite expected value.
\begin{definition}[First Entrance Time]
	Let $\mathcal{M}=(I,X_0,P)$ be a Markov chain and $j\in I$ be a state. We call \emph{first entrance time} into $j$ the random variable with values in $\mathbb{N}\cup \{+\infty\}$
	$$
	T_j= \begin{cases}
	\min\{n\geq 1|X_n=j\} & \textnormal{if }\{n\geq 1|X_n=j\}\neq \emptyset,\\
	+\infty & \textnormal{otherwise}.
	\end{cases}
	$$
\end{definition}
We show now how to derive a Markov chain $\mathcal{M_P}=(I,\cdot,P)$ from any FPARS $\mathcal{P}=(A,\mathsf{P})$. First we define a relation $\equiv_\bot$ on the states of $A$ defined for each $a,b\in A$ as
$$
a\equiv_\bot b\text{ if and only if }a,b\in\mathbf{NF}(A)\text{ or }a=b.
$$
Clearly $\equiv_\bot$ is an equivalence relation. We call $A_{\equiv\bot}$ the \emph{quotient set} of $A$ by $\equiv_\bot$ and $\mathsf{trm}$ the \emph{equivalence class} of $a$ if $a\in\mathbf{NF}(A)$. We define $\mathcal{M_P}=(I,\cdot,P)$ in the following way.
$$
I=A_{\equiv\bot},\qquad\qquad
p_{ij}=\begin{cases}
1 & \text{if } i,j\in\mathbf{NF}(A),\\
\mathbb{P}(i\rightarrow j) & \text{otherwise}.\\	
\end{cases}	
$$
Actually, we have defined not one but a \emph{family} of Markov chains, since 
the 
initial distribution is not specified. When we talk about $\mathcal{M_P}$ we 
are implicitly universally quantifying over all the initial laws, which are 
countable once one only considers initial configurations which are Dirac.

It is not obvious that the AST and PAST properties we defined are actually 
meaningful. Now we prove their relationship with classical properties defined 
for Markov chains (as AST and PAST were defined in \cite{bournez_proving_2005, 
ferrer_fioriti_probabilistic_2015}), in a way similar to what have been done in 
\cite{avanzini_probabilistic_2018} for a slightly different model.
\begin{remark}
	Given an FPARS $\mathcal{P}$ and a sequence of configurations $(\rho_i)_{i\in\mathbb{N}}$, from how we have defined $\mathcal{M_P}$, it follows that $|\rho_k| = \mathbb{P}\{T_\mathsf{trm}\geq k\}$.
\end{remark}
\begin{proposition}
	An FPARS $\mathcal{P}$ is AST if and only if $\mathbb{P}\{T_{\mathsf{trm}}<+\infty\}=1$, for the Markov chain $\mathcal{M_P}$.
\end{proposition}
\begin{proof}
	An FPARS $\mathcal{P}$ is AST if and only if $\underset{n\rightarrow +\infty}{\lim}\left|\rho_{n}\right|=0$. Since $|\rho_n| = \mathbb{P}\{T_\mathsf{trm}\geq n\}$, then $\mathbb{P}\{T_\mathsf{trm}< n\}=1-|\rho_n|$. Taking the limit at both sides, we have $\underset{n\rightarrow +\infty}{\lim} \mathbb{P}\{T_\mathsf{trm}< n\} = \mathbb{P}\{T_\mathsf{trm}< +\infty\}= \underset{n\rightarrow +\infty}{\lim} (1-|\rho_n|)=1-\underset{n\rightarrow+\infty}{\lim} |\rho_n|$. Thus $\mathbb{P}\{T_\mathsf{trm}< +\infty\}=1$ if and only if $\underset{n\rightarrow+\infty}{\lim} |\rho_n|=0.$\qed
\end{proof}
\begin{proposition}
	Given an FPARS $(A,\mathsf{P})$, $\rho_0=\mathbf{Dirac}(a)$,
	where $a\in A$ and the computation $(\rho_i)_{i\in\mathbb{N}}$, $\mathbb{E} [ T_{\mathsf{trm}}]=\sum\limits_{i=1}^{\infty} \left|\rho_i\right|=\nsteps{\mathsf{P}}(a)$.
\end{proposition}
\begin{proof}
	By Proposition~\ref{prop:alt-def}, $\mathbb{E} [ T_{\mathsf{trm}}]=\sum\limits_{i=1}^{\infty} \mathbb{P}\{T_{\mathsf{trm}}\geq i\}=\sum\limits_{i=1}^{\infty} \left|\rho_i\right|=\nsteps{\mathsf{P}}(a)$.\qed
\end{proof}

\subsection{Related Works}
ARSs were first considered with probabilistic strategies without nondeterminism in \cite{bournez_probabilistic_2002} and then with nondeterminism in \cite{bournez_proving_2005}. In these works AST and PAST properties were introduced, as a rewriting analogue of recurrence and positive recurrence in Markov chain literature. This framework was refined in \cite{avanzini_probabilistic_2018}, which introduced multidistributions to handle nondeterminism in a simple and uniform way. We were inspired by this latter work and used a similar terminology, in particular partial probability distributions that admit their mass to sum to less then one. However, while in \cite{avanzini_probabilistic_2018} PARSs are first-class citizens, we attach probabilities to a pre-existing ARSs, in order to solve nondeterminism. For this reason we have reformulated their model, getting rid of multidistributions, thus coming up with a very light and simple framework.
%%%%%%%%%%%%%%%%%%%%%%%%%%%%%%%%%%%%%%%%%%%%%%%%%%%%%%%%%
\section{Randomised Strategies in the $\lambda$-calculus}\label{section:prostra}
%%%%%%%%%%%%%%%%%%%%%%%%%%%%%%%%%%%%%%%%%%%%%%%%%%%%%%%%%
In the previous Sections we have defined all the mathematical machinery needed 
for defining randomised strategies for ARSs, in the abstract. In this Section 
we focus on the $\lambda$-calculus, as defined in Section \ref{section:basic}, 
as the target of our investigation. In deterministic strategies, \emph{the} 
redex which is being reduced is typically chosen according to its position in 
the term. In randomised strategies we have more freedom. Intuitively, we have 
to assign a probability to each redex of any term, making them sum to one. The 
space of possible choices is indeed very large. A first design choice could be 
the answer to the following question: should \emph{every} redex in a term being 
reduced with strictly positive probability? If the answer is positive, then one 
should decide how to define these numbers.

%%%%%%%%%%%%%%%%%%%%%%%%%%%%%%%%%
\subsection{The Uniform Strategy}
%%%%%%%%%%%%%%%%%%%%%%%%%%%%%%%%%

Maybe the most trivial way in which one could assign probabilities to redexes is in a \emph{uniform} way.

\begin{definition}[Uniform Strategy]
	 $\psuni$ is a randomised reduction strategy for the ARS $(\Lambda,\redbeta)$ such that for each reducible term $\termone$, $\psuni(\termone)=\mu\in\dist{\Lambda}$, where for each $\termtwo\in\Lambda$:
	 $$
	 \mu(\termtwo)=\frac{|\rdxs{\termone\red\termtwo}|}{|\rdxs{\termone}|} 
	 $$
	where $\rdxs{\termone}$ is the set of redex occurrences of $\termone$, $\rdxs{\termone\red\termtwo}=\{\rdxone\in\rdxs{\termone}|\termone\redbeta\termtwo\textrm{ by firing }\rdxone\}$ and $|\cdot|$ denotes the cardinality of a set.
	
\end{definition}
Although very easy to \emph{define}, the uniform strategy is not easy at all to 
\emph{study}. In fact, the number of redexes in terms along a reduction 
sequence 
can grow. As we have seen in the previous Section, we study randomised 
strategies as FPARSs. In the case of the $\lambda$-calculus, we focus our 
attention on FPARSs in the form $(\Lambda_\mathsf{WN},\cdot)$, because we are 
interested in studying quantitative properties about terminating terms. Whether 
$(\Lambda_\mathsf{WN},\psuni)$ is AST is still an open problem, but for sure 
$(\Lambda_\mathsf{WN},\psuni)$ is \emph{not} PAST, as we are going to prove 
with the following counterexample. 
\begin{example}
	Let us  consider the combinator $\termone=(\lambda x.\lambda y.y)\bm{\Delta_4}$, where $\bm{\Delta_4}=\bm{\Delta_2\Delta_2}$ and $\bm{\Delta_2}= \lambda x.(xx)(xx)$. We observe that $\bm{\Delta_4}$ reduces to $\bm{\Delta_4\Delta_4}$. Thus, if we do not reduce leftmost, the number of redexes grows linearly with the length of the reduction sequence. Instead, reducing leftmost, we reach the normal form in one step. Considering the computation $|(\rho_i)|_{i\in\mathbb{N}}$ starting from $\rho_0=\textbf{Dirac}(\termone)$ (in Figure \ref{figure:delta2}) we can compute the sequence $|\rho_i|$.
	\begin{equation*}
	\begin{split}
		|\rho_0|&=1\,,\qquad\qquad\qquad\qquad\,\,|\rho_1|=1\,,\\
		|\rho_2|&=|\rho_1|-|\rho_1|\frac{1}{2}\,,\qquad\qquad
		|\rho_3|=|\rho_2|-|\rho_2|\frac{1}{3}\,.
	\end{split}
	\end{equation*}
	In general, if $i\geq 2$, it holds that
	$$
	|\rho_i|=|\rho_{i-1}|-|\rho_{i-1}|\frac{1}{i}=|\rho_{i-1}|\left( 
	1-\frac{1}{i}\right).
	$$
	Once solved the recurrence relation in the last line, one has 
	$|\rho_i|=\frac{1}{i}$ for each $i\geq 1$. Thus
	$$
	\nsteps{\psuni}(\termone)=\sum\limits_{i=1}^{\infty} \left|\rho_i\right|=\sum\limits_{i=1}^{\infty} \frac{1}{i}=+\infty.
	$$
\end{example}
Since there is a term $\termone\in\Lambda_\mathsf{WN}$ such that $\nsteps{\psuni}(\termone)=+\infty$, the following holds.
\begin{proposition}
	$(\Lambda_\mathsf{WN},\psuni)$ is not PAST.
\end{proposition}
\tikzset{
itria/.style={
	draw,dashed,shape border uses incircle,
	isosceles triangle,shape border rotate=90,yshift=-1.27cm},
}
\begin{figure}
\begin{center}
	\tikzstyle{level 1}=[level distance=1.5cm, sibling distance=2.5cm]
	\tikzstyle{level 2}=[level distance=1.7cm, sibling distance=4cm]
	\tikzstyle{level 3}=[level distance=1.1cm, sibling distance=4cm]
	\fbox{
		\begin{minipage}{.96\textwidth}
			\tikzstyle{level 3}=[level distance=0.5cm, sibling distance=2.5cm]
			\centering
			\begin{tikzpicture}[grow=down]
			\node[term] {$(\lambda x.\lambda y.y)\bm{\Delta_4}$}
			child {
				node[term] {$\lambda y.y$}        
				edge from parent         
				node[left]  {$\frac{1}{2}$}
			}
			child {
				node[term] {$(\lambda x.\lambda y.y)(\bm{\Delta_4\Delta_4})$}        
				child {
					node[term] {$\lambda y.y$}        
					edge from parent         
					node[above]  {$\frac{1}{3}$}
				}
				child {
					node[term] {$(\lambda x.\lambda y.y)((\bm{\Delta_4}\bm{\Delta_4})\bm{\Delta_4})$}
					child{ node[draw=none] {}
						{ node[itria] {$\cdots$} } 
					}
					edge from parent         
					node[right]  {$\frac{1}{3}$}
				}
				child {
					node[term] {$(\lambda x.\lambda y.y)(\bm{\Delta_4}(\bm{\Delta_4\Delta_4}))$}
					child{ node[draw=none] {}
						{ node[itria] {$\cdots$} } 
					}
					edge from parent         
					node[above]  {$\frac{1}{3}$}
				}
				edge from parent         
				node[right]  {$\frac{1}{2}$}
			};
			\end{tikzpicture}
		\end{minipage}}
	\end{center}
	\caption{The tree representing the reduction of term $\termone = (\lambda x.\lambda y.y)\bm{\Delta_4}$ under randomised strategy $\psuni$.}
	\label{figure:delta2}
\end{figure}
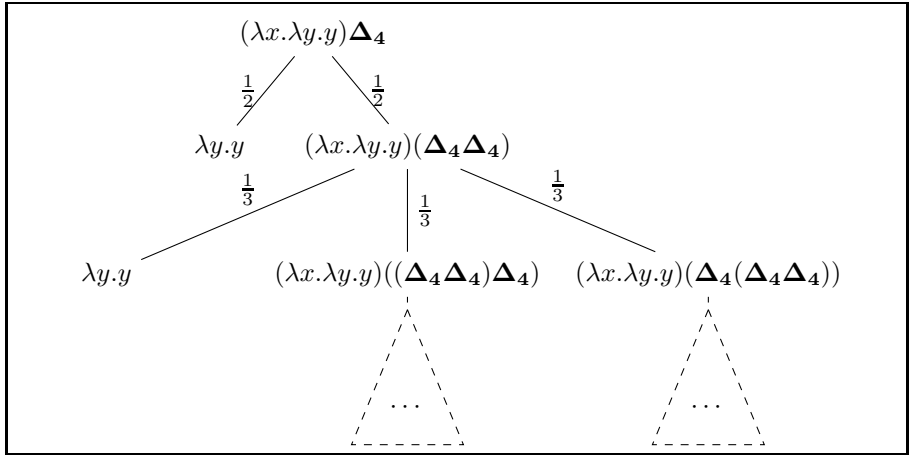
\subsection{The Strategy $P_\varepsilon$}
In real programming languages we would like to have a normalising strategy, or 
at least, since we are in probabilistic context, a (positive) almost-surely 
normalising one. For this reason, we devised a new strategy for the ARS 
$(\Lambda,\redbeta)$, simpler than the uniform one, which mixes 
leftmost-outermost and rightmost-innermost strategies. If 
$\varepsilon\in[0,1]$, call $\mathsf{P}_\varepsilon$ the strategy in which the 
leftmost redex is always chosen with probability $\varepsilon$ and the 
rightmost one with probability $1-\varepsilon$.

\begin{definition}[Strategy $\mathsf{P}_\varepsilon$]
	Given a reducible term $\termone$ and $\varepsilon\in[0,1]$, the randomised strategy $\mathsf{P}_\varepsilon$ is such that $\mathsf{P}_\varepsilon(\termone)=\mu\in\dist{\Lambda}$, where for each $\termtwo\in\Lambda$:
	$$
	\mu(\termtwo)=\begin{cases}
	\varepsilon & \text{if }\termone\redbeta\termtwo\text{ by firing the }\pslo\text{-redex,}\\
	1-\varepsilon & \text{if }\termone\redbeta\termtwo\text{ by firing the }\psri\text{-redex,}\\
	0 & \text{otherwise}.\\
	\end{cases}
	$$
\end{definition}

\begin{notation}
	For readability reasons we define for the FPARS $(\Lambda,\mathsf{P}_\varepsilon)$ the family of functions $\explen{\termone}:[0,1]\rightarrow\mathbb{R}\cup\{+\infty\}$ indexed on a term $\termone\in\Lambda$, where
	$\explen{\termone}(\varepsilon)=\nsteps{\mathsf{P}_\varepsilon}(\termone)$.
\end{notation}
\begin{example}
	Let us consider the term $\termone=(\lambda x.y)\bm{\Omega}$ where $\bm{\Omega}=\bm{\omega\omega}$  with $\bm{\omega}=\lambda x.xx$. There are two possible representations of the behaviour of $\mathsf{P}_\varepsilon$ for this term, one as an infinite tree (Figure \ref{figure:trees}a) and another one as a cyclic graph (Figure \ref{figure:trees}b). According to the different representations, we can compute in different ways the probability of reaching normal form and the average derivation length. The results coincide yielding in both cases probability of termination equal to $1$ and average derivation length equal to $\frac{1}{\varepsilon}$, if $\varepsilon\neq 0$. If $\varepsilon=0$, the system never terminates and thus the average derivation length is $+\infty$. 
	%Moreover the reader can note the analogy with the abstract setting of Example \ref{example:pars}.
\end{example}
\begin{figure}[t]
	\begin{center}
		\fbox{
			\begin{minipage}{.96\textwidth}
				\begin{center}
					\tikzstyle{level 1}=[level distance=1.2cm, sibling distance=2.5cm]
					\tikzstyle{level 2}=[level distance=1.2cm, sibling distance=2.5cm]
					\subfloat[]{{\begin{tikzpicture}[grow=down]
							
							\node[term] {$(\lambda x.y)\bm{\Omega}$}
							child {
								node[term] {$y$}        
								edge from parent         
								node[left]  {$\varepsilon$}
							}
							child {
								node[term] {$(\lambda x.y)\bm{\Omega}$}        
								child {
									node[term] {$y$}        
									edge from parent         
									node[left]  {$\varepsilon$}
								}
								child {
									node[term,label=below:{\begin{rotate}{-90}{$\cdots$}\end{rotate}}] {$(\lambda x.y)\bm{\Omega}$}
									edge from parent         
									node[right]  {$1-\varepsilon$}
								}
								edge from parent         
								node[right]  {$1-\varepsilon$}
							};
							\end{tikzpicture}}}
					\qquad\qquad\qquad
					\tikzstyle{level 1}=[level distance=2.1cm, sibling distance=2.5cm]
					\subfloat[]{{\begin{tikzpicture}[grow=down]
							\node[term]  (M) {$(\lambda x.y)\bm{\Omega}$}
							child {
								node[term] {$y$}        
								edge from parent         
								node[left]  {$\varepsilon$}
							};
							\path[-]
							(M) edge [loop above] node {$1-\varepsilon$} ();
							
							\end{tikzpicture}}}
				\end{center}
			\end{minipage}}
		\end{center}
		\caption{The tree (a) and the cyclic graph (b) representing the reduction sequence of $\termone$.}
		\label{figure:trees}
	\end{figure}
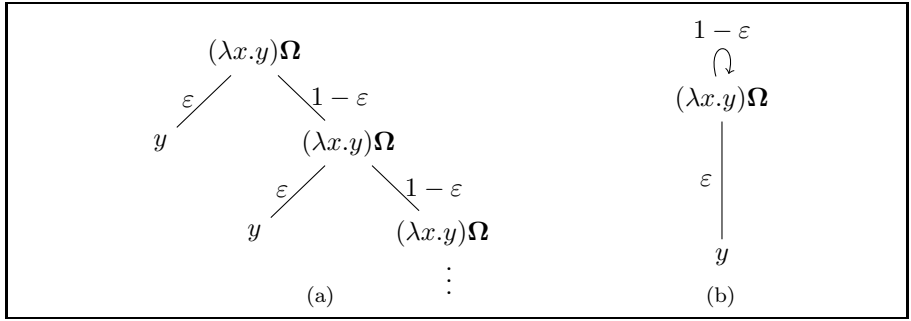
The first Theorem we are going to prove about $\mathsf{P}_\varepsilon$ states 
that $\mathsf{P}_\varepsilon$ is an almost-surely normalising strategy. The key 
aspect of this proof is that the $\pslo$-redex has always the \emph{same} 
probability $\varepsilon$ of being reduced, along the whole reduction sequence.
\begin{theorem}
	The FPARS $(\Lambda_\mathsf{WN},\mathsf{P_\varepsilon})$ is PAST whenever $\varepsilon>0$.
\end{theorem}
\begin{proof}
	We use Foster's Theorem to prove the claim. Thus, we have to find a suitable Lyapunov function $f:\Lambda_\mathsf{WN}\rightarrow\mathbb{R}$.
	We consider $f=\nsteps{\pslo}$. Certainly, condition (1) is verified since $\nsteps{\pslo}\left(\termone\right)\geq0$ for each $\termone\in\Lambda_\mathsf{WN}$. We have to verify (2). Suppose $\mathsf{P_\varepsilon}(\termone)=\mu$. If $\termone\redlo\termtwo$ and $\termone\redri\termthree$, by Lemma~\ref{lemma:derivationlength} we can write:
	\begin{equation*}
	\begin{split}
	\nsteps{\pslo}\left( \mu\right) & = \nsteps{\pslo}(\termtwo) \cdot \varepsilon+ \nsteps{\pslo}(\termthree)\cdot (1-\varepsilon)\\
	& \leq \left( \nsteps{\pslo}\left( \termone\right)-1\right)  \cdot \varepsilon + \nsteps{\pslo}(\termone)\cdot (1-\varepsilon)\\
	& = \varepsilon\cdot\nsteps{\pslo}\left( \termone\right)-\varepsilon + \nsteps{\pslo}\left( \termone\right)\cdot\left( 1-\varepsilon\right)\\
	& = \nsteps{\pslo}\left( \termone\right)-\varepsilon.
	\end{split}
	\end{equation*}
	Since $0\leq\varepsilon\leq 1$, $\nsteps{\pslo}\left(M\right)>_{\varepsilon}\nsteps{\pslo}\left(\mu\right)$ for each normalising term $\termone$.
	Then, if $\varepsilon>0$, $(\Lambda_\mathsf{WN},\mathsf{P_\varepsilon})$ is PAST and the average number of steps to normal form of a term $\termone$ reduced with strategy $\mathsf{P}_\varepsilon$ is bounded by
	$\frac{\nsteps{\pslo}\left( \termone\right)}{\varepsilon}$.\qed
\end{proof}
The bound we obtain on $\explen{\termone}(\varepsilon)$ from the above proof is very loose and thus it does not give us any information on the actual nature of the function $\explen{\termone}(\varepsilon)$. We show, by means of an example, that the strategy $\mathsf{P}_\varepsilon$ is non-trivial i.e. there exists a term $\termone$ and $0<\varepsilon<1$, such that $\explen{\termone}(\varepsilon)<\explen{\termone}(1)=\nsteps{\pslo}(\termone)<\explen{\termone}(0)=\nsteps{\psri}(\termone)$.
\begin{example}
	Let us consider a family of terms $\termone_n=\termtwo\termthree_n$ where:
	$$
	\termtwo = \lambda x.\underbrace{((\lambda y.z)\bm{\Omega})}_\termfour x\qquad
	\termthree_n = C_n\underbrace{((\lambda x.x)y)}_\termfive\qquad
	C_n = \lambda  x.\underbrace{xx\cdots x}_{n\text{ times}}
	$$
	After quite simple computations one can derive $\explen{\termone_n}(\varepsilon)=(n-3)\varepsilon^3+4\varepsilon^2+\frac{2}{\varepsilon}$. Clearly for $\varepsilon=0$ the expression diverges. If $n\geq 2$ there is a minimum for $0<\varepsilon<1$, and thus $\explen{\termone_n}(\varepsilon)<\explen{\termone_n}(1)=\nsteps{\pslo}(\termone_n)<\explen{\termone_n}(0)=\nsteps{\psri}(\termone_n)=+\infty$. $\explen{\termone_2}(\varepsilon)$ is plotted in Figure \ref{figure:plot}.
	\begin{figure}[t]
          \fbox{
            \begin{minipage}{.96\textwidth}
		\centering
		\begin{tikzpicture}[transform shape,scale=0.7]
		\begin{axis}[
		xmin = 0, xmax = 1,
		ymin = 4, ymax = 10,
		y label style={at={(axis description cs:0.11,.5)},anchor=south},
		domain=0.1:1,
		samples=99,
		no markers,
		line width=1pt,
		xlabel=$\varepsilon$,
		ylabel={$\explen{\termone_2}(\varepsilon)$}
		] 
		\addplot {2*x^(-1) + 4*x^2 - x^3}; 
		\end{axis}
		\end{tikzpicture}
            \end{minipage}}
		\caption{The function $\explen{\termone_2}(\varepsilon)$.}
		\label{figure:plot}
	\end{figure}
\end{example}
Studying the behaviour of $\explen{\termone}(\varepsilon)$ for \emph{an arbitrary} term $\termone$ is
a difficult task, which is outside the scope of this paper. However, we are 
able to characterise 
$\explen{\termone}:[0,1]\rightarrow\mathbb{R}\cup\{+\infty\}$ from an 
analytical point of view, e.g. investigating critical points and convexity. 
This is particularly interesting since we are interested in the efficiency of 
$\mathsf{P_\varepsilon}$. Since $\mathsf{P_0}$ is just $\psri$ and 
$\mathsf{P_1}$ is $\pslo$, a global minimum (respectively, maximum) of 
$\explen{\termone}$ strictly between $0$ and $1$ would mean that reducing under 
$\mathsf{P_\varepsilon}$ could lead to normal form in less (respectively, 
more), steps on average, than reducing deterministically under either $\pslo$ 
or $\psri$. Call $\poly{x}$ the set of polynomials in the unknown $x$ with 
coefficients in $\mathbb{R}$.
\begin{lemma}\label{lemma:poly}
	Every configuration $\rho$ of the FPARS $\left(\Lambda,\mathsf{P}_{\varepsilon}\right)$ is such that $\rho(\termone)\in \poly{\varepsilon}$ for each $\termone\in\Lambda$.
\end{lemma}
\begin{proof}
	We prove this by induction on the length $\ell$ of the computation $\rho_0\rightsquigarrow\rho_1\rightsquigarrow\rho_2\rightsquigarrow\cdots\rho$. If $\ell=0$, $\rho=\rho_0$, thus $\rho(\termone)=1$ if $\rho_0=\textbf{Dirac}(\termone)$, $0$ otherwise. Now let us suppose that for computations with $\ell\leq k$, $\rho(\termone)\in \poly{\varepsilon}$. Let us consider a computation $ \rho_0\rightsquigarrow\rho_1\rightsquigarrow \rho_2\rightsquigarrow\cdots\rho_{k}\rightsquigarrow\rho$ of length $k+1$. $\rho(\termone)=\underset{\termtwo\in\Lambda}{\sum}\rho_{k}\left( \termtwo\right) \cdot\mathbb{P}\left(\termtwo \rightarrow \termone\right)$. Polynomials are closed with respect to addition and multiplication, thus by induction hypothesis $\rho(\termone)\in\poly{\varepsilon}$ if $\mathbb{P}(\termtwo \rightarrow \termone)\in\poly{\varepsilon}$. This is certainly true because $\mathbb{P}(\termtwo \rightarrow \termone)$ is either $0,\varepsilon,\textrm{ or }(1-\varepsilon)$ from the definition of $\mathsf{P}_{\varepsilon}$.\qed
\end{proof}
\begin{theorem}
	For each $\lambda$-term $\termone$, the function $\explen{\termone}(\varepsilon)$ is a power series.
\end{theorem}
\begin{proof}
	Let us consider the FPARS $\mathcal{P}=\left(\Lambda,\mathsf{P}_{\varepsilon}\right)$ and a computation $(\rho_i)_{i\in\mathbb{N}}$ starting from $\rho_0=\textbf{Dirac}(\termone)$. $\supp{\rho_i}$
	is finite for every $i\geq 0$, since the number of redexes in a term is 
	finite and thus $|\rho_i|\in\poly{\varepsilon}$ for every $i\geq 0$, since 
	it is a sum of polynomials by Lemma~\ref{lemma:poly}. Therefore 
	$\explen{\termone}(\varepsilon)=\sum\limits_{i=1}^{\infty} |\rho_i|$ is a 
	series whose terms are polynomials in $\varepsilon$, which, once the terms 
	are reordered, is a power series.\qed
\end{proof}
%%%%%%%%%%%%%%%%%%%%%%%%%%%%%%%%%%%%%%%%%%%%%%%
\subsection{Optimality and Pessimality Results}\label{sect:optimalpessimal}
%%%%%%%%%%%%%%%%%%%%%%%%%%%%%%%%%%%%%%%%%%%%%%%
As a further step in the direction of a full understanding of the nature
of $\explen{\termone}(\varepsilon)$, we study it in the special
case in which $\termone$ is a term of the subcalculi $\lambda A$ and $\lambda 
I$ we have previously introduced. In particular, we prove that 
$\explen{\termone}(\varepsilon)$ has minimum (respectively maximum)
in $\varepsilon=1$ for $\lambda A$-terms (respectively for $\lambda I$-terms).
All we need to do is to lift Theorem~\ref{theorem:optpes}, a result about 
\emph{deterministic}
strategies, to the randomised setting. Some preliminary lemmas are necessary in
order to appropriately do so.

The following two lemmas tell us that the existence of a strictly partial
probability distribution along a computation witnesses the existence
of a \emph{deterministic} computation leading to normal form.
\begin{lemma}\label{lemma:exseq}
Let $(A, \mathsf{P})$ an FPARS and $(\rho_i)_{i\in\mathbb{N}}$ a
computation, where $\rho_0=\mathbf{Dirac}(a_0)$. For each $a\in A$, if there
exists $k\geq 0$ such that $\rho_k(a)>0$, then there exists a
reduction sequence $a_0\rightarrow a_1\rightarrow\cdots\rightarrow
a_{k-1}\rightarrow a$.
\end{lemma}
\begin{proof}
	We argue by induction on $k$. If $k=0$, then the reduction sequence is 
	trivially $a_0=a$. If $k=h$, $\rho_h(a)=\underset{b\in 
	A}{\sum}\rho_{h-1}\left( b\right) \cdot\mathbb{P}\left(b \rightarrow 
	a\right)$. Since $\rho_h(a)>0$, there exists $b\in A$ such that 
	$\rho_{h-1}\left( b\right) \cdot\mathbb{P}\left(b \rightarrow a\right)\neq 
	0$, i.e. $\rho_{h-1}(b)>0$ and $\mathbb{P}\left(b \rightarrow a\right)>0$. 
	Thus, by induction hypothesis, there exists a sequence $a_0\rightarrow 
	a_1\rightarrow \cdots\rightarrow a_{h-2}\rightarrow b$, and $b \rightarrow 
	a$. Hence there exists a reduction sequence $a_0\rightarrow a_1\rightarrow 
	\cdots\rightarrow a_{h-2}\rightarrow b\rightarrow a$.\qed
\end{proof}
\begin{lemma}\label{lemma:exseq2}
	Let $(A, \mathsf{P})$ an FPARS and $(\rho_i)_{i\in\mathbb{N}}$ a computation, where $\rho_0=\mathbf{Dirac}(a_0)$. If there exists $k\geq 1$ such that $|\rho_k|<1$, then there exists a sequence $a_0\rightarrow a_1\rightarrow\cdots\rightarrow a_{j}$ such that $a_{j}$ is in normal form and $j\leq k-1$.
\end{lemma}
	\begin{proof}
		We argue by induction on $k$. If $k=1$, since $|\rho_0|$ is $\mathbf{Dirac}(a_0)$ then $\mathsf{P}(a_0)$ is undefined (otherwise $|\rho_1|=1$). Hence $a_0$ is in normal form. If $k=h$ and $|\rho_{h-1}|<1$ by induction hypothesis we are done. So let us consider the case in which $|\rho_{h-1}|=1$ and $|\rho_h|<1$. We claim that there exists $b\in \mathbf{NF}(A)$ such that $\rho_{h-1}(b)>0$.
		\begin{align*}
			|\rho_h|&=\underset{a\in A}{\sum}\,\underset{b\in A}{\sum}\rho_{h-1}(b) \cdot\mathbb{P}\left(b \rightarrow a\right)\\
			&=\underset{b\in A}{\sum}\,\underset{a\in A}{\sum}\rho_{h-1}(b) \cdot\mathbb{P}\left(b \rightarrow a\right)=\underset{b\in A}{\sum}\left(\rho_{h-1}(b)\underset{a\in A}{\sum}\mathbb{P}\left(b \rightarrow a\right) \right)\\
			&=\underset{b\not\in \mathbf{NF}(A)}{\sum}\left(\rho_{h-1}(b)\underset{a\in A}{\sum}\mathbb{P}\left(b \rightarrow a\right) \right)+\underset{b\in \mathbf{NF}(A)}{\sum}\left(\rho_{h-1}(b)\underset{a\in A}{\sum}\mathbb{P}\left(b \rightarrow a\right) \right).
		\end{align*}
		If there was not $b\in \mathbf{NF}(A)$ such that $\rho_{h-1}(b)>0$, then the second term in the sum would vanish and
		$|\rho_h|=\underset{b\not\in \mathbf{NF}(A)}{\sum}\rho_{h-1}(b)=1$.
		But $|\rho_h|<1$ by hypothesis. Hence there exists $b\in \mathbf{NF}(A)$ such that $\rho_{h-1}(b)>0$ and thus by Lemma \ref{lemma:exseq} there exist a sequence $a_0\rightarrow a_1\rightarrow\cdots\rightarrow a_{h-2}\rightarrow b$.\qed
	\end{proof}
We are almost done: the following lemma tells us that all
configurations along a computation starting from a $\lambda A$-term $\termone$
are proper until the $n$-th configuration, where
$n=\nsteps{\pslo}(\termone)$.
\begin{lemma}\label{lemma:min}
  Given the FPARS $(\Lambda_A, \mathsf{P_\varepsilon})$, and a
  computation $(\rho_i)_{i\in\mathbb{N}}$, where
  $\rho_0=\mathbf{Dirac}(\termone_0)$, for each
  $k\leq\nsteps{\pslo}(\termone_0)$, then $|\rho_k|=1$.
\end{lemma}
\begin{proof}
	Let $n=\nsteps{\pslo}(\termone_0)$. If there were $k\leq n$ such that $|\rho_k|<1$, then by Lemma \ref{lemma:exseq2} we could find a sequence $\termone_0\redbeta\termone_1\redbeta\cdots\redbeta\termone_j$ such that $j\leq k-1$ and $\termone_j$ is normal form. But this is impossible because of Theorem \ref{theorem:optpes}.\qed
\end{proof}
\begin{corollary}
	For each term $\termone$ in $\Lambda_{A}$, $\explen{\termone}(\varepsilon)$ has minimum in $\varepsilon=1$.
\end{corollary}
\begin{proof}
	Let $n=\nsteps{\pslo}(\termone)=\explen{\termone}(1)$.
	\begin{align*}
		\explen{\termone}(\varepsilon)&=\sum\limits_{i=1}^{\infty} |\rho_i|=\sum\limits_{i=1}^{n} |\rho_i|+\sum\limits_{i=n+1}^{\infty} |\rho_i|\overset{\textrm{Lemma \ref{lemma:min}}}{=}\sum\limits_{i=1}^{n}1+\sum\limits_{i=n+1}^{\infty} |\rho_i|\\
		&=n+\sum\limits_{i=n+1}^{\infty}|\rho_i|=\explen{\termone}(1)+\sum\limits_{i=n+1}^{\infty}|\rho_i|\geq \explen{\termone}(1).
	\end{align*} \qed
\end{proof}
Dually, the following Lemma shows that all configurations $\rho_k$ along a computation starting from a $\lambda I$-term $\termone$ are null (i.e. $|\rho_k|=0$), if $k>n=\nsteps{\pslo}(\termone)$. 
\begin{lemma}\label{lemma:max}
	Given the FPARS $(\Lambda_I, \mathsf{P_\varepsilon})$, and a
	computation $(\rho_i)_{i\in\mathbb{N}}$, where
	$\rho_0=\mathbf{Dirac}(\termone_0)$, for each
	$k>\nsteps{\pslo}(\termone_0)$, it holds that $|\rho_k|=0$.
\end{lemma}
\begin{proof}
	Let $n=\nsteps{\pslo}(\termone_0)$. If there were $k>n$ such that $|\rho_k|>0$, then by Lemma \ref{lemma:exseq} we could find a reduction sequence $\termone_0\redbeta\termone_1\redbeta\cdots\redbeta\termone_{k}$. But this is impossible because of Theorem~\ref{theorem:optpes}.\qed
\end{proof}
\begin{corollary}
	For each term $\termone$ in $\Lambda_{I}$, $\explen{\termone}(\varepsilon)$ has maximum in $\varepsilon=1$.
\end{corollary}
\begin{proof}
	Let $n=\nsteps{\pslo}(\termone)=\explen{\termone}(1)$. If $n=+\infty$ we are done. Then we consider $n$ finite.
	\begin{align*}
	\explen{\termone}(\varepsilon)&=\sum\limits_{i=1}^{\infty} |\rho_i|=\sum\limits_{i=1}^{n} |\rho_i|+\sum\limits_{i=n+1}^{\infty} |\rho_i|\overset{\textrm{Lemma \ref{lemma:max}}}{=}\sum\limits_{i=1}^{n} |\rho_i|\leq n= \explen{\termone}(1).
	\end{align*} \qed
\end{proof}
%%%%%%%%%%%%%%%%%%%%%%%%%%%%%%%%%%%%%%%%%%%%%%%%%
\subsection{Further Results}
%%%%%%%%%%%%%%%%%%%%%%%%%%%%%%%%%%%%%%%%%%%%%%%%%
We have seen in the previous Section that in particular cases
$\explen{}(\varepsilon)$ has maximum or minimum in
$\varepsilon=1$. However, we know nothing about about the shape of the
curve $\explen{}(\varepsilon)$. In this Section we are going to show
some counterintuitive results about the different shapes
$\explen{}(\varepsilon)$ can have in $\lambda A$, $\lambda I$
and in the full $\lambda$-calculus. In particular, we address the study of 
convexity and critical points. Indeed we are interested in minima and maxima, 
i.e. those points where the randomised strategy is maximally efficient and, 
respectively, inefficient. 

\begin{definition}[Critical Point]
	Let $f:A\subseteq\mathbb{R}\rightarrow\mathbb{R}$. $x\in A$ is a \emph{critical point} of $f$ if $f$ is not derivable in $x$ or the derivative of $f$ in $x$ is $0$.
\end{definition}

First of all, the results in Section~\ref{sect:optimalpessimal}
may suggest that $\explen{\termone}(\cdot)$ could be monotonically
decreasing in $\Lambda_A$, and monotonically increasing in $\Lambda_I$.
This is actually \emph{not the case}, and the counterexamples
are precisely the terms $\termone_A$ and $\termone_I$ we already
considered in Section~\ref{section:basic}, namely:
$$
\termone_I=(\lambda x.x\bm{I})(\lambda x.(\lambda z.zz)(xy))
\qquad
\termone_A=(\lambda x.x\bm{I})(\lambda x.(\lambda z.y)(xy))
$$
The graphs of the two functions $\explen{\termone_I}(\cdot)$
and $\explen{\termone_A}(\cdot)$ are in Figure~\ref{figure:counterexample2}.
As a consequence:
\begin{proposition}
  $\explen{\termone}(\cdot)$ is not monotonic, neither in
  $\Lambda_A$ nor in $\Lambda_I$.
\end{proposition}
\begin{figure}
	\fbox{
		\begin{minipage}{.96\textwidth}
			\centering
			\subfloat[]{{
					\begin{tikzpicture}[transform shape,scale=0.6]
					\begin{axis}[
					xmin = 0, xmax = 1,
					ymin = 4.5, ymax = 5.5,
					y label style={at={(axis description cs:0.11,.5)},anchor=south},
					domain=0:1,
					samples=100,
					no markers,
					line width=1pt,
					xlabel=$\varepsilon$,
					ylabel={$\explen{\termone_I}(\varepsilon)$}
					] 
					\addplot {5 - x^2 + x^3}; 
					\end{axis}
					\end{tikzpicture}
			}}
			\subfloat[]{{
					\begin{tikzpicture}[transform shape,scale=0.6]
					\begin{axis}[
					xmin = 0, xmax = 1,
					ymin = 2.5, ymax = 3.5,
					y label style={at={(axis description cs:0.11,.5)},anchor=south},
					domain=0:1,
					samples=100,
					no markers,
					line width=1pt,
					xlabel=$\varepsilon$,
					ylabel={$\explen{\termone_A}(\varepsilon)$}
					] 
					\addplot {3 + x^2 - x^3}; 
					\end{axis}
					\end{tikzpicture}
			}}
		\end{minipage}}
	\caption{Plot of $\explen{\termone_I}(\varepsilon)$ (a) and $\explen{\termone_A}(\varepsilon)$ (b).}
	\label{figure:counterexample2}
\end{figure}
Looking at the terms $\termone_I$ and $\termone_A$, one immediately
notices that the functions we are considering not only fail to be
monotonic in general, but that they fail to be \emph{convex}
or \emph{concave}. Is this common to all terms? If one considers
terms like
$$
\termone_\cup=((\lambda y.z)(\bm{II}))((\lambda x.xx)(\bm{I}y))
\qquad
\termone_\cap=((\lambda x.xx)(\bm{I}y))((\lambda y.z)(\bm{II}))
$$
one immediately realizes that $\explen{\termone}(\varepsilon)$ can
indeed be concave or concave in certain cases (see Figure
\ref{figure:examples}(a) and Figure \ref{figure:examples}(b) for a
plot). Playing a bit with terms allowed us to find terms in
which $\explen{\termone}(\cdot)$ can be quite wild, having
more than one critical point. Consider, as an example, the term
$$
\termone_\sim=((\lambda y.\termone_\cap)(\bm{II}))((\lambda x.xx)(\bm{I}y))
$$
and the plot of $\explen{\termone_\sim}(\cdot)$, reported in
Figure \ref{figure:examples}(c).
\begin{proposition}
  $\explen{\termone}(\varepsilon)$ can have more than one critical
  point.
\end{proposition}
\begin{figure}[t]
	\fbox{
		\begin{minipage}{.96\textwidth}
	          \begin{center}
			\subfloat[]{{
				\begin{tikzpicture}[transform shape,scale=0.6]
				\begin{axis}[
				xmin = 0, xmax = 1,
				ymin = 2, ymax = 6,
				y label style={at={(axis description cs:0.11,.5)},anchor=south},
				domain=0:1,
				samples=100,
				no markers,
				line width=1pt,
				xlabel=$\varepsilon$,
				ylabel={$\explen{\termone_\cup}(\varepsilon)$}
				] 
				\addplot {4 - 3*x + 4*x^2 - x^3}; 
				\end{axis}
				\end{tikzpicture}
			}}
			\subfloat[]{{
					\begin{tikzpicture}[transform shape,scale=0.6]
					\begin{axis}[
					xmin = 0, xmax = 1,
					ymin = 2, ymax = 6,
					y label style={at={(axis description cs:0.11,.5)},anchor=south},
					domain=0:1,
					samples=100,
					no markers,
					line width=1pt,
					xlabel=$\varepsilon$,
					ylabel={$\explen{\termone_\cap}(\varepsilon)$}
					] 
					\addplot {4 + 3*x - 3*x^2 + x^3 - x^4}; 
					\end{axis}
					\end{tikzpicture}
			}}
                  \end{center}
                  \begin{center}
				\subfloat[]{{
						\begin{tikzpicture}[transform shape,scale=0.6]
						\begin{axis}[
						xmin = 0, xmax = 1,
						ymin = 6, ymax = 10,
						y label style={at={(axis description cs:0.11,.5)},anchor=south},
						domain=0:1,
						samples=100,
						no markers,
						line width=1pt,
						xlabel=$\varepsilon$,
						ylabel={$\explen{\termone_\sim}(\varepsilon)$}
						] 
						\addplot {8 - 3*x + 21*x^2 - 56*x^3 + 75*x^4 - 76*x^5 + 59*x^6 - 21*x^7 + x^8}; 
						\end{axis}
						\end{tikzpicture}
				}}
                  \end{center}
	\end{minipage}}
		\caption{Plot of $\explen{\termone_\cup}(\varepsilon)$ (a), $\explen{\termone_\cap}(\varepsilon)$ (b) and $\explen{\termone_\sim}(\varepsilon)$ (c).}
		\label{figure:examples}
	\end{figure}

%%%%%%%%%%%%%%%%%%%%%%%%%%%%%%%%%%%%%%%%%%%%%%%%%
\section{Conclusions}\label{section:conclusions}
In this work we have started the study of randomised reduction
strategies for the $\lambda$-calculus. We have defined a family of
examples of such strategies, and we have shown that all of them,
except one, are positive almost-surely normalising. Then we have
studied how those strategies behave in $\lambda A$ (the affine
$\lambda$-calculus) and $\lambda I$, proving optimality and pessimality results.
Moreover, we have shown that our defined
family of strategies behaves in a very complex way in the scope of the full $\lambda$-calculus.

Further work could consist in better analysing the behaviour of the proposed
strategies, in particular trying to characterize classes of $\lambda$-terms for which our strategies
work strictly better than deterministic ones, and to develop some methods
to tune the parameter $\varepsilon$ in order to get good performances.
Moreover, it would be interesting to study the behaviour of randomised
strategies on non-terminating $\lambda$-terms, investigating the perpetuality 
phenomenon.
%

%\section*{References}
\bibliographystyle{elsarticle-num}
\bibliography{TCS}
\end{document}